\documentclass[10pt, conference]{IEEEtran}

\usepackage{algorithm}
\usepackage{algpseudocode}
\usepackage{amsfonts,amsmath,amssymb,amsthm}
\usepackage{bm}
\usepackage{color,epsfig,graphics,graphicx}
\usepackage{epstopdf}
\usepackage{dsfont}
\usepackage{float}
\usepackage{times}
\usepackage{caption}
\usepackage{subcaption}
\usepackage{tikz}
\usepackage{cite,url}
\usepackage{psfrag}

\usetikzlibrary{shapes,arrows}

\newcommand{\commentout}[1]{}

\newtheorem{theorem}{Theorem}

\newtheorem{lemma}{Lemma}
\newtheorem{proposition}{Proposition}

\renewcommand{\qed}{\rule{5pt}{5pt}}

\newcommand{\cL}{\mathcal{L}}

\newcommand{\eps}{\varepsilon}

\newcommand{\abs}[1]{\left|#1\right|}

\newcommand{\E}[2]{\mathbb{E}_{#1} \! \left[#2\right]}

\newcommand{\normw}[2]{\left\|#1\right\|_{#2}}

\tikzstyle{block} = [rectangle, draw, text width=5em, text centered, rounded corners, minimum height=4em]
\tikzstyle{inv} = [rectangle, text width=5em, text centered, minimum height=4em]
\tikzstyle{line} = [draw, -latex']

\begin{document}

\title{Managing your Private and Public Data: \\
Bringing down Inference Attacks against your Privacy}

\author{
\IEEEauthorblockN{Salman Salamatian\IEEEauthorrefmark{1}, Amy Zhang\IEEEauthorrefmark{2}, Flavio du Pin Calmon\IEEEauthorrefmark{3}, Sandilya Bhamidipati\IEEEauthorrefmark{2},\\
Nadia Fawaz\IEEEauthorrefmark{2}, Branislav Kveton\IEEEauthorrefmark{2}, Pedro Oliveira\IEEEauthorrefmark{2}, Nina Taft\IEEEauthorrefmark{2}\\
}
\IEEEauthorblockA{\IEEEauthorrefmark{1}
EPFL,
Lausanne, Switzerland,
Email: salman.salamatian@epfl.ch\\}
\IEEEauthorblockA{\IEEEauthorrefmark{2}
Technicolor,
Palo Alto, CA 94301,\\
Email: \{amy.zhang, sandilya.bhamidipati, nadia.fawaz, branislav.kveton, nina.taft\}@technicolor.com\\}
\IEEEauthorblockA{\IEEEauthorrefmark{3}
MIT,
Cambridge, MA 02138,
Email: flavio@mit.edu\\}
}

\IEEEoverridecommandlockouts
\IEEEpubid{\makebox[\columnwidth]{Parts of this technical report were presented in IEEE GlobalSIP 2013 \cite{salamatian2013globalsip}} \hspace{\columnsep}\makebox[\columnwidth]{ }}

\maketitle

\begin{abstract}

We propose a practical methodology to protect a user's private data, when he wishes to publicly release data that is correlated with his private data, in the hope of getting some utility.
Our approach relies on a general statistical inference framework that captures the privacy threat under inference attacks, given utility constraints. Under this framework, data is distorted before it is released, according to a privacy-preserving probabilistic  mapping. This mapping is obtained by solving a convex optimization problem, which minimizes information leakage under a distortion constraint.
We address practical challenges encountered when applying this theoretical framework to real world data. On one hand, the design of optimal privacy-preserving mechanisms requires knowledge of the prior distribution linking private data and data to be released, which is often unavailable in practice. On the other hand, the optimization may become untractable and face scalability issues when data assumes values in large size alphabets, or is high dimensional. Our work makes three major contributions.
First, we provide bounds on the impact on the privacy-utility tradeoff of a mismatched prior.
Second, we show how to reduce the optimization size by introducing a quantization step, and how to generate privacy mappings under quantization. 
Third, we evaluate our method on three datasets, including a new dataset that we collected, showing correlations between political convictions and TV viewing habits. We demonstrate that good privacy properties can be achieved with limited distortion so as not to undermine the original purpose of the publicly released data, e.g. recommendations.
\end{abstract} 
\section{Introduction}\label{sec:Introduction}

In recent years, the many dangers of online privacy abuse have surfaced, including identity theft, reputation loss, job
loss, discrimination, harassment, cyberbullying, stalking and even suicide \cite{cyberbullying,suicide,reputation}.  During the same time, many highly visible privacy lawsuits have burst on the scene that typically accuse online social network (OSN) providers of not properly informing users about what their data is used for and whom else gets access to it.  We have seen lawsuits on illegal data collection \cite{googlesuit1}, sharing data without user consent  \cite{googlebuzz}, changing privacy settings without informing users \cite{facebooksuit1}, misleading users about tracking their browsing behavior \cite{fbtracking}, not carrying out user deletion actions \cite{europe-vs-facebook}, and more. The potential cost of losing these law suits is rising into the tens and hundreds of millions of dollars \cite{FBcost,GoogleFine}. These events beg for academics to focus more on bridging the divide between theoretical privacy and practical issues of implementation.

One of the central problems of managing privacy in the Internet lies in the
simultaneous management of both public and private data. Many users are willing
to release \emph{some} data about themselves, such as their movie watching
history or their gender; they do so because such data enables useful services
and because such attributes are rarely considered private. However users also
have other data they consider private, such as income level, political
affiliation, or medical conditions. In this work, we focus on a method in
which a user can release her public data, but is able to prevent against
inference attacks that may learn her private data from the public information. Our solution consists of a privacy-preserving mapping, which informs a user on
how to distort her public data, before releasing it, such that no inference
attacks can successfully learn her private data. At the same time, the
distortion should be bounded so that the original service (such as a
recommendation) can continue to be useful.

In this paper we adopt the privacy framework presented in
\cite{alerton2012privacy}. This general framework considers the privacy threat incurred by a user
when a passive adversary attempts to infer the user's private information from
the user's public (released) data. The privacy loss is measured in terms of an inference
cost gain that the adversary has by observing the released data. The goal of the
framework is to determine a mapping of the public data to a new set of
outputs given certain distortion (utility) constraints. The authors in \cite{alerton2012privacy}
formulate the problem of determining this mapping  for a
general inference cost function as a convex program. Without
significant loss of generality, \cite{alerton2012privacy} argues that the
privacy loss can be measured in terms of mutual information, which leads to an
optimization formulation similar to the one found in rate-distortion theory. This formulation,
albeit general and theoretically sound, faces a number of practical challenges
when applied to actual datasets available within web services.
\IEEEpubidadjcol
The first challenge is that this method relies on knowing a joint probability distribution between the private and public data, called the \emph{prior}. Often the true prior distribution is not available and instead only a limited set of samples of the private and public data can be observed. This leads to the \emph{mismatched prior} problem. We seek to provide a meaningful distortion and bring privacy even in the face of a mismatched prior. Our first contribution centers around this. Starting with the set of observable data samples, we find an improved estimate of the prior, based on which the privacy-preserving mapping is derived. We develop some bounds on any additional distortion this process incurs to guarantee a given level of privacy. More precisely, we show that the private information leakage increases log-linearly with the $\cL_1$-norm distance between our estimate and the prior; that the distortion rate increases linearly with the $\cL_1$-norm distance between our estimate and the prior; and that the $\cL_1$-norm distance between our estimate and the prior decreases as the sample size increases.

The second challenge is one of scalability that occurs when the size of the underlying alphabet of the user data is very large, e.g. due to a large number of features representing the data. To handle this, we propose a quantization approach that limits the dimensionality of the problem. We preprocess and quantize the original data by clustering it. We then determine how to distort the data in the space defined by the clusters. The privacy-preserving mapping is computed using a convex solver that minimizes privacy leakage subject to a distortion constraint. The advantage of our quantization scheme is that it is computationally efficient - we reduce the number of optimized variables from being quadratic in the size of the underlying feature alphabet to being quadratic in the number of clusters, and thus make the optimization independent of the number of observable data samples. For some real world examples, this can lead to orders of magnitude reduction in dimensionality. We also show that any additional distortion introduced by quantization increases linearly with the maximum distance between a sample datapoint and the closest cluster center. This quantization step, our second contribution, provides a fundamental extension to the original method in \cite{alerton2012privacy} which sometimes can not easily be applied in practice when the data is too high dimensional.

Our third area of contribution centers around evaluations. In \cite{alerton2012privacy} the authors only proposed and reasoned about their framework but did not evaluate it. To the best of our knowledge, our paper is the first evaluation of this method. We evaluate our methods on 3 datasets, 2 well known datasets and one new dataset that we collected ourselves. This latter dataset is one that contains users TV show ratings and their political affiliation. The Simmons National Consumer Survey organization in the US has shown that correlations exist between political affiliation and opinions about TV shows \cite{Experian-survey}. In this case study, we consider TV show opinions to be data to be released and a user's political affiliation to be kept private. The general framework of privacy against statistical inference \cite{alerton2012privacy} allows for different kinds of instantiations of data distortions before public release. For example, our system might suggest to a user to simply remove an element of their public data (called \emph{erasure-distortions}), or may suggest to alter the contents of some elements in a public profile (called \emph{exchange-distortions}), or other forms of distortion.

Our evaluations demonstrate multiple things. First, even when we do not have a fully specified prior distribution on the public and private distribution, we can still provide privacy. We show that we can provide privacy in this difficult environment at the extra cost of a small amount of additional distortion in the public data. Second, we illustrate that our quantization approach works well, namely that it is possible to provide good privacy even when quantization is needed to reduce the dimensionality of the data. Third, we show that in our Politics-and-TV dataset, perfect privacy can be achieved with a 15\% distortion of the original public data. In practice less than 15\% distortion could provide sufficient privacy. We also illustrate examples of specific distortions (changes to particular public data profiles) and show these are intuitively reasonable, yet not trivial.

The rest of this paper is organized as follows. In the next section we formally define the problem at hand. In Section~\ref{sec:mismatchedPrior} we provide bounds on the privacy-distortion tradeoff when the mismatched prior problem surfaces. In Section~\ref{sec:highDim}, we explain our method of quantization to address scalability challenges. Our datasets are described in Section~\ref{sec:dataSets}, and the results of our evaluation are provided in Section~\ref{sec:Results}. We discuss related work in Section~\ref{sec:relatedWork}, and conclude in Section~\ref{sec:Conclusion}.

\section{Problem Statement}\label{sec:Problem}

In this section, we define the threat model, and describe the privacy-accuracy framework considered in this paper. Then, we point out two challenges encountered when applying this framework in practice, and outline our approaches to address these challenges. These approaches are treated in more details in Sections~\ref{sec:mismatchedPrior} and~\ref{sec:highDim}.

\subsection{Notations}

We denote by $\rm{Simplex}$ the probability simplex defined by $\sum_{x} p(x) = 1, \; \; p(x) \geq 0 \;\forall x $.
Let $A\in \mathcal{A}$ and $B\in \mathcal{B}$ be random vectors taking values in the alphabets $\mathcal{A}$ and $\mathcal{B}$ respectively. The joint probability distribution of the elements of $A$ and $B$ is denoted $p_{A,B}: \mathcal{A} \times \mathcal{B} \rightarrow [0,1]$. The marginal distribution of vector $A$ is defined by $p_A (a) = \sum_{b \in \mathcal{B}} p_{A,B} (a,b)  \; \; \forall a \in \mathcal{A}$, while the conditional distribution of $A$ given $B$ is given by $p_{A | B} (a | b) = \frac{p_{A,B} (a,b)}{p_B(b)}$, where support issues are handled accordingly. We may drop the subscripts when the interpretation is clear from the context.

We recall the definitions of the entropy $H(A)$ of a random vector $A$, and of the mutual information $I(A;B)$ of vectors $A$ and $B$:
\begin{align}
\label{eq:entropy}
& H(A)= - \sum_{a \in \mathcal{A}} p_A(a) \log(p_A(a)) \nonumber \\
& I(A;B)= \sum_{a \in \mathcal{A}, b \in \mathcal{B}} p_{A,B}(a,b) \log \left(
 \frac{p_{A,B}(a,b)}{p_{A}(a) p_B(b)}\right).
\end{align}
Note that $H(A)$ depends on the distribution of $A$ only, while $I(A;B)$ depends only on the joint distribution $p_{A,B}$ of $A$ and $B$, since the marginals $p_A$ and $p_B$ can be obtained from the joint distribution.

\subsection{Threat Model}\label{sec:ThreatModel}

We consider the setting described in \cite{alerton2012privacy},  where a user has two types of data: some data that he would like to remain private, e.g. his income level, his political views, etc.,  and some data that he is willing to release publicly and from which he will derive some utility, for example the release of his media preferences to a service provider would allow the user to receive content recommendations.
We denote by $A \in \mathcal{A}$ the vector of personal attributes that the user wants to keep private, and by $B \in \mathcal{B}$ the vector of data he is willing to make public, where $\mathcal{A}$ and $\mathcal{B}$ are the sets from which $A$ and $B$ can assume values.

We assume that the user private attributes $A$ are linked to his data $B$ by the joint probability distribution $p_{A,B}$. Thus, an adversary who would observe $B$ could infer some information about $A$ from $B$.

To reduce this inference threat, instead of releasing $B$, the user will release a \emph{distorted version} of $B$, denoted $\hat{B} \in \mathcal{\hat{B}}$, generated according to a conditional probabilistic mapping $p_{\hat{B}|B}$, called the \emph{privacy-preserving mapping}. Note that the set $\mathcal{\hat{B}}$ may differ from the set $\mathcal{B}$.

The privacy-preserving mapping $p_{\hat{B}|B}$ should be designed in such a way that it renders any statistical inference of $A$ based on the observation of $\hat{B}$ harder, yet, at the same time, preserves some utility to the released data $\hat{B}$, by limiting the distortion generated by the mapping. This can be modeled by a constraint $\Delta \geq 0$ on the average distortion:
\begin{equation}\label{eq:distortionConstraint}
E_{B,\hat{B}}[d(B,\hat{B})] \leq \Delta,
\end{equation}
for some distortion metric $d:\mathcal{B} \times \mathcal{\hat{B}} \rightarrow \mathds{R}^+$. It should be noted that any distortion metric can be used, such as the Hamming distance if $B$ and $\hat{B}$ are binary vectors, or the $l_2$-norm if $B$ and $\hat{B}$ are real vectors, or even more complex metrics modeling the variation in utility that a user would derive from the release of $\hat{B}$ instead of $B$. The latter could, for example, represent the difference in the quality of content recommended to the user based on the release of his distorted media preferences $\hat{B}$ instead of his true preferences~$B$.


We now formalize the privacy threat model. 
We assume the following standard statistical inference threat model \cite{alerton2012privacy}: 
the adversary chooses a belief on the data $A$, modeled by a probability distribution $q:\mathcal{A}\rightarrow [0,1]$. The belief $q$ is obtained by minimizing an expected cost function $C(A,q)$. 
In particular, prior to observing $\hat{B}$, the adversary chooses his inference method (i.e. belief) $q$ as the solution of the minimization
$$
c^*_0 = \min_{q} E_A[C(A,q)].
$$
After observing $\hat{B}$, the adversary updates his inference method $q$ such that it minimizes
$$
c^*_{\hat{b}} = \min_{q} E_{A|\hat{B}}[C(A,q)|\hat{B}=\hat{b}].
$$
The average cost gain by the adversary after observing the public release $\hat{B}$ is the difference
$$
\Delta C = c^*_0 - E_{\hat{B}} [c^*_{\hat{b}}].
$$
This average cost gain represents how much an adversary gains in term of inference of the private attributes $A$ thanks to the observation of $\hat{B}$. The goal of the privacy-preserving mapping will be to minimize this gain. In the particular case of \emph{perfect privacy} $\Delta C=0$, the released data $\hat{B}$ does not provide any information that is helpful for the inference of $A$, and the inference cannot outperform an uninformed guess. Note that this general framework does not assume a particular inference algorithm.

If an adversary uses the log-loss\footnote{For a justification of the relevance and generality of the log-loss cost, we refer the reader to \cite[Section~IV.A]{alerton2012privacy}.} cost function $C(A,q)=-\log(q_{A})$, it can easily be shown \cite{alerton2012privacy} that
\begin{equation}
\Delta C = I(A;\hat{B}).
\end{equation}
Hence, the privacy leakage is captured by the mutual information between the private attributes $A$ and the publicly released data $\hat{B}$. It should be noted that in the case of perfect privacy $(I(A;\hat{B})=0)$, the privacy-preserving mapping $p_{\hat{B}|B}$ renders the released data $\hat{B}$ statistically independent from the private data $A$.

It should be mentioned that, although we model the privacy threat using the average cost gain $\Delta C$ in this paper, Calmon and Fawaz \cite{alerton2012privacy} also proposed a worst-case model $\Delta C^*= c_0^*-\min_{\hat{b}\in \mathcal{B}} c^*_{\hat{b}}$, where the privacy threat is measured in terms of the most informative output, i.e. the output that gives the largest gain in cost.  We would like to point out that in the case of perfect privacy under the log-loss, the average threat model $\Delta C=0$ and the worst-case threat model $\Delta C^*=0$ are equivalent. Thus conclusions drawn on distortion to achieve perfect privacy under the average threat model also hold for the worst-case model. In general, the worst-case threat is an upperbound on the average threat, and its analysis and application are the object of some of our ongoing work.

\subsection{Privacy-Accuracy Framework}\label{sec:PrivacyAccuracy}

\begin{algorithm}[t]
  \caption{Privacy preserving mapping.}
  \label{alg:PPM}
  \begin{algorithmic}
    \State {\bf Input:} prior $p_{A, B}$
    \State
    \State solve the problem for $p_{\hat{B} | B}$:
\begin{align*}
  \underset{p_{\hat{B} | B}}{\text{minimize}} & \quad
  J(p_{A, B}, p_{\hat{B} | B}) \\
  \text{subject to} & \quad
  \E{p_{B, \hat{B}}}{d(B, \hat{B})} \leq \Delta \\
  & \quad p_{\hat{B} | B} \in \rm{Simplex}
\end{align*}
    \State
    \State {\bf Output:} mapping $p_{\hat{B} | B}$
  \end{algorithmic}
\end{algorithm}

In this section, we describe how the privacy-preserving mapping is designed to address the inference privacy threat, under a constraint on the distortion.

The mutual information $I(A;\hat{B})$ is a function of the joint distribution $p_{A,\hat{B}}$, which in turn depends on both the prior distribution $p_{A,B}$ and the privacy-preserving mapping $p_{\hat{B}|B}$. Indeed, $A \rightarrow B \rightarrow \hat{B}$ form a Markov chain, thus
\begin{align}\label{eq:distribution}
p_{A,\hat{B}}(a, \hat{b}) = & \sum_{b \in \mathcal{B}} p_{\hat{B}|B}(\hat{b} | b) p_{A,B}(a, b),\nonumber\\
p_{\hat{B}} (\hat{b}) = & \sum_{b \in \mathcal{B}} p_{\hat{B}|B} (\hat{b} | b) p_{B}(b),
\end{align}
and using Eq. (\ref{eq:distribution}) in the definition of $I(A;\hat{B}) $, we can write
\begin{equation}
\begin{split}
I(A;\hat{B})
&=\sum_{a, b,\hat{b}} p_{A,B}(a,b)p_{\hat{B}|B}(\hat{b}|b) \log \frac{\sum_{b"} p(\hat{b}|b") p(b"|a)}{\sum_{a',b'} p(\hat{b}|b') p(a',b')}.
\end{split}
\end{equation}
To stress the dependency of the privacy leakage on the prior distribution and the privacy-preserving mapping, we will denote $$I(A;\hat{B})=J(p_{A,B}, p_{\hat{B}|B}).$$
Similarly, the average distortion $E_{B,\hat{B}}[d(B,\hat{B})]$ is a function of the joint distribution $p_{B,\hat{B}}$, which in turn depends both on the prior distribution $p_{A,B}$, through the marginal $p_B$, and on the privacy-preserving mapping $p_{\hat{B}|B}$.

Consequently, given a prior distribution $p_{A,B}$ linking the private attributes $A$ and the data $B$, the privacy-preserving mapping $p_{\hat{B}|B}$ minimizing the privacy leakage subject to a distortion constraint is obtained as the solution to the optimization problem
\begin{equation}
	\begin{aligned}
	& \underset{p_{\hat{B}|B}}{\text{minimize}}
	& &  J(p_{A,B},p_{\hat{B}|B})\\
	& \text{subject to}
	& & E_{B,\hat{B}}[d(B,\hat{B})] \leq \Delta \\
	& & & p_{\hat{B}|B} \in \rm{Simplex}.
	\end{aligned}
	\label{eq: convex opt}
\end{equation}
It was shown in \cite{alerton2012privacy} that this problem is convex, and can thus be efficiently solved using standard algorithms. Note that this problem bears some resemblance with a modified rate distortion problem. This optimization is summarized in Algorithm~\ref{alg:PPM}.

\subsection{Practical Challenges}

In this section, we describe two practical challenges encountered when applying the theoretical privacy-accuracy framework described in Section~\ref{sec:PrivacyAccuracy}.

\textbf{Mismatched prior}: Finding the privacy-preserving mapping as the solution to the convex optimization in Algorithm~\ref{alg:PPM} relies on the fundamental assumption that the prior distribution $p_{A,B}$ that links private attributes $A$ and data $B$ is known and can be fed as an input to the algorithm.
In practice, the true prior distribution may not be known, but may rather be estimated from a set of sample data that can be observed, for example from a set of users who do not have privacy concerns and publicly release both their attributes $A$ and their original data $B$. The prior estimated based on this set of samples from non-private users is then used to design the privacy-preserving mechanism that will be applied to new users, who are concerned about their privacy. In practice, there may exist a mismatch between the estimated prior and the true prior, due for example to a small number of observable samples, or to the incompleteness of the observable data.
In Section~\ref{sec:mismatchedPrior}, we characterize the actual privacy-accuracy tradeoff that results from first running Algorithm~\ref{alg:PPM} with a mismatched prior as input, and then using the so-obtained privacy-preserving mapping, instead of the mapping that would have been obtained under the knowledge of the true prior. \\

\textbf{Large Data}: Designing the privacy-preserving mapping $p_{\hat{B}|B}$ requires characterizing the value of $p_{\hat{B}|B}(\hat{b}|b)$ for all possible pairs $(b,\hat{b})\in \mathcal{B} \times \mathcal{\hat{B}}$, i.e. solving the convex optimization problem over $|\mathcal{B}| |\mathcal{\hat{B}}|$ variables. When $\mathcal{\hat{B}}=\mathcal{B}$, and the size of the alphabet $|\mathcal{B}|$ is large, solving the convex optimization over $|\mathcal{B}|^2$ variables may become intractable.
In Section~\ref{sec:highDim}, we propose a method based on quantization to reduce the number of optimization variables. We show that this method to reduce complexity does not affect the privacy levels that can be achieved,  but comes at the expense of a limited amount of additional distortion, that we characterize.

\section{Privacy in the face of mismatched prior}\label{sec:mismatchedPrior}
\label{sec:mismatch}

Suppose that we do not have  perfect knowledge of the true prior distribution $p_{A, B}$ but that we have its estimate $q_{A, B}$. Then, if $q_{A, B}$ is a good estimate of $p_{A, B}$, the solution $p_{\hat{B} | B}^\ast$ obtained by feeding the mismatched distribution $q_{A, B}$ as an input to the optimization problem~\eqref{eq: convex opt} should be close to the one with $p_{A, B}$. In particular, the information leakage $J(q_{A, B}, p_{\hat{B} | B}^\ast)$ and distortion due to the mapping $p_{\hat{B} | B}^\ast$,  with respect to the mismatched prior $q_{A, B}$ should be similar to the actual leakage $J(p_{A, B}, p_{\hat{B} | B}^\ast)$ and distortion with respect to the true prior $p_{A, B}$. This claim is formalized in the following theorem.

\begin{theorem}
\label{thm:prior mismatch}
Let $p_{\hat{B} | B}^\ast$ be a solution to the optimization problem~\eqref{eq: convex opt} with $q_{A, B}$. Then:
\begin{align*}
  & \abs{J(p_{A, B}, p_{\hat{B} | B}^\ast) - J(q_{A, B}, p_{\hat{B} | B}^\ast)} \\
  & \quad \leq 3 \normw{p_{A, B} - q_{A, B}}{1} \log
  \frac{\abs{\mathcal{A}} \abs{\mathcal{B}}}{\normw{p_{A, B} - q_{A, B}}{1}} \\
  & \E{p_{\hat{B}, B}}{d(\hat{B}, B)} \leq \Delta + d_\mathrm{max} \normw{p_{A, B} - q_{A, B}}{1}
\end{align*}
where $d_\mathrm{max} = \max_{\hat{b}, b} d(\hat{b}, b)$ is the maximum distance in the feature space.
\end{theorem}

The following lemma \cite{Cover}, which bounds the difference in the entropies of two distributions,  will be useful in the proof of Theorem~\ref{thm:prior mismatch}.
\begin{lemma}[{\cite[Thm~17.3.3]{Cover}}]
\label{lemma}
Let $p$ and $q$ be distributions with the same support $\mathcal{X}$ such that $\normw{p - q}{1} \leq \frac{1}{2}$. Then:
\begin{align*}
  \abs{H(p) - H(q)} \leq \normw{p - q}{1} \log \frac{\abs{\mathcal{X}}}{\normw{p - q}{1}}.
\end{align*}
\end{lemma}

\noindent {\bf Proof of Theorem~\ref{thm:prior mismatch}: }
\noindent The first inequality can be proved in four steps. Initially, we note that the objective function can be rewritten as
\begin{align}
  J(p_{A, B}, p_{\hat{B} | B}) = H(p_{A}) + H(p_{\hat{B}}) - H(p_{A, \hat{B}}).
\end{align}
Therefore, the difference between the objective functions with respect to $p_{A, B}$ and $q_{A, B}$ is bounded as:
\begin{align}
  & \abs{J(p_{A, B}, p_{\hat{B} | B}) - J(q_{A, B}, p_{\hat{B} | B})}
  \label{eq:mismatch entropy} \\
  & \quad \leq |H(p_{A}) - H(q_{A})| \ + \nonumber \\
  & \quad \quad \ |H(p_{\hat{B}}) - H(q_{\hat{B}})| \ + \nonumber \\
  & \quad \quad \ |H(p_{A, \hat{B}}) - H(q_{A, \hat{B}})|. \nonumber
\end{align}
The bound in Lemma~\ref{lemma} can be used to bound each of the terms in Equation~\eqref{eq:mismatch entropy}. For instance:
\begin{align}
  \normw{p_{A, \hat{B}} - q_{A, \hat{B}}}{1}
  \ = & \ \sum_{a, \hat{b}} \abs{\sum_b p(\hat{b} | b) [p(a, b) - q(a, b)]} \nonumber \\
  \ \leq & \ \sum_{a, b, \hat{b}} p(\hat{b} | b) \abs{p(a, b) - q(a, b)} \nonumber \\
  \ = & \ \sum_{a, b} \underbrace{\sum_{\hat{b}} p(\hat{b} | b)}_1 \abs{p(a, b) - q(a, b)} \nonumber \\
  \ = & \ \normw{p_{A, B} - q_{A, B}}{1}
\end{align}
and therefore:
\begin{align}
  & |H(p_{A, \hat{B}}) - H(q_{A, \hat{B}})| \label{eq:mismatch H1} \\
  & \quad \leq \normw{p_{A, B} - q_{A, B}}{1} \log
  \frac{\abs{\mathcal{A}} \abs{\mathcal{B}}}{\normw{p_{A, B} - q_{A, B}}{1}}.
  \nonumber
\end{align}
Similarly, it can be shown that:
\begin{align}
  & |H(p_{A}) - H(q_{A})| \label{eq:mismatch H2} \\
  & \quad \leq \normw{p_{A, B} - q_{A, B}}{1} \log
  \frac{\abs{\mathcal{A}}}{\normw{p_{A, B} - q_{A, B}}{1}} \nonumber \\
  & |H(p_{\hat{B}}) - H(q_{\hat{B}})| \label{eq:mismatch H3} \\
  & \quad \leq \normw{p_{A, B} - q_{A, B}}{1} \log
  \frac{\abs{\mathcal{B}}}{\normw{p_{A, B} - q_{A, B}}{1}}. \nonumber
\end{align}
Finally, the three upper bounds can be substituted into Equation~\eqref{eq:mismatch entropy}, which yields:
\begin{align}
  & \abs{J(p_{A, B}, p_{\hat{B} | B}) - J(q_{A, B}, p_{\hat{B} | B})} \\
  & \quad \leq
  3 \normw{p_{A, B} - q_{A, B}}{1} \log
  \frac{\abs{\mathcal{A}} \abs{\mathcal{B}}}{\normw{p_{A, B} - q_{A, B}}{1}}.
  \nonumber
\end{align}
Our first claim is proved by substituting $p_{\hat{B} | B}^\ast$  for $p_{\hat{B} | B}$ in the above equation.

The proof of our second claim is based on the following inequality:
\begin{align}
  & \abs{\E{p_{\hat{B}, B}}{d(\hat{B}, B)} - \E{q_{\hat{B}, B}}{d(\hat{B}, B)}} \nonumber \\
  & \quad =
  \abs{\sum_{a, b, \hat{b}} p(\hat{b} | b) [p(a, b) - q(a, b)] d(b, \hat{b})} \nonumber \\
  & \quad \leq
  \sum_{a, b, \hat{b}} p(\hat{b} | b) d(b, \hat{b}) \abs{p(a, b) - q(a, b)} \nonumber \\
  & \quad \leq d_\mathrm{max} \sum_{a, b}
  \underbrace{\sum_{\hat{b}} p(\hat{b} | b)}_1 \abs{p(a, b) - q(a, b)} \nonumber \\
  & \quad = d_\mathrm{max} \normw{p_{A, B} - q_{A, B}}{1}.
\end{align}
Based on this observation, it follows that:
\begin{align}
  \E{p_{\hat{B}, B}}{d(\hat{B}, B)} \
  \leq & \ \E{q_{\hat{B}, B}}{d(\hat{B}, B)} + \nonumber \\
  & \ d_\mathrm{max} \normw{p_{A, B} - q_{A, B}}{1} \nonumber \\
  \leq & \ \Delta + d_\mathrm{max} \normw{p_{A, B} - q_{A, B}}{1}.
\end{align}
The last step is due to the constraint $\E{q_{\hat{B}, B}}{d(\hat{B}, B)} \leq \Delta$ that is enforced in our problem~\eqref{eq: convex opt}. This concludes our proof. \qed

\bigskip

\noindent Finally, we provide a bound on the probability of $\normw{p_{A, B} - q_{A, B}}{1}$ being large, when $q_{A,B}$ is simply the empirical distribution obtained from counting on $n$ samples.

\begin{proposition}
Let $q_{A,B}$ be empirical distribution of $p_{A, B}$:

\begin{align*}
q_{A,B}(a,b)= \frac{\#\{a_i=a, b_i=b\}}{n}
\end{align*}
where $n$ is the total number of samples, and  $\#\{a_i=a, b_i=b\}$ is the number of examples where $A = a$ and $B = b$. Then

\begin{align*}
\mathds{P}(\|q_{A,B}-p_{A,B}\|_1 \geq \eps) \leq (n+1)^{|\mathcal{A}||\mathcal{B}|} 2^{-{2n\epsilon^2}}
\end{align*}
\end{proposition}
\noindent {\bf Proof:}
By Pinsker's Inequality, we get :

\begin{align}
\epsilon \leq  \|p_{A,B} - q_{A,B} \|_1 \leq \sqrt{ \frac{1}{2} D(p_{A,B} \| q_{A,B}) } 
\end{align}
which implies that $D(p_{A,B} \| q_{A,B}) \geq 2 \epsilon^2$. We now combine this with Sanov's Theorem to directly obtain the desired bound:

\begin{align}
\mathds{P}(\|q_{A,B}-p_{A,B}\|_1 > \eps) & \leq (n+1)^{|\mathcal{A}||\mathcal{B}|} 2^{-D(p_{A,B}\|q_{A,B})} \\
& \leq (n+1)^{|\mathcal{A}||\mathcal{B}|} 2^{-{2n\epsilon^2}}
\end{align}

\qed

\bigskip

\noindent Therefore, as the sample size $n$ increases, the probability of having a poor empirical estimator of the true distribution in terms of $\cL_1$-norm decreases with rate $(n+1)^{\mathcal{|A}||\mathcal{B}|}2^{-{2n\epsilon^2}}$

\section{Privacy for Large Data}\label{sec:highDim}

\begin{figure*}
\centering
\begin{tikzpicture}[node distance = 3cm, auto]
\node [inv] (data) {Discrete data};
\node [block, right of=data] (estimation) {Probability estimation};
\node [block, right of=estimation] (quantization) {Quantization};
\node [block, right of=quantization] (opt) {Convex Optimization};
\node [inv, right of=opt] (mapping) {Optimal Privacy Mapping};

\path [line] (data) -- (estimation);
\path [line] (estimation) -- node {$p_{A,B}$} (quantization);
\path [line] (quantization) -- node {$p_{A,C}$} (opt) ;
\path [line] (opt) -- node {$p_{\hat{C}|C}$} (mapping) ;
\end{tikzpicture}
\caption{The quantization approach for large alphabets}
\label{fig:diag}
\end{figure*}
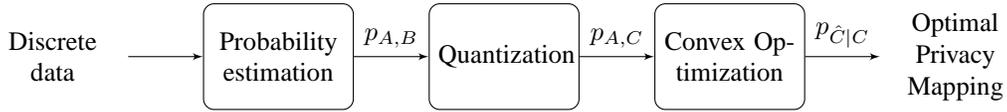

In real-world datasets, the alphabet $\mathcal{B}$ is often large. In particular, the number of symbols in the alphabet~$\mathcal{B}$ observed in the available dataset may be  $\theta(n)$, linear in the number of samples $n$ in the dataset. Suppose that $\hat{\mathcal{B}} = \mathcal{B}$. Then the number of optimized variables in Problem~\eqref{eq: convex opt} is  $\theta(n^2)$. 
Note that the distortion constraint is linear in~$p_{\hat{B} \mid B}(\hat{b} \mid b)$ , but the objective function is neither linear nor quadratic.  As a result, the optimization problem~(\ref{eq: convex opt}) cannot be solved using fast linear or quadratic programming solvers. In general, the problem is hard to solve when the size of alphabet $\mathcal{B}$ exceeds a few hundreds symbols.


To address this issue, we show how to solve our problem approximately by optimizing fewer variables. Our method comprises three steps. 
First, a quantization  \cite{gray98quantization} step maps the symbols in alphabet $\mathcal{B}$ to $\abs{\mathcal{C}}$ representative examples in a smaller alphabet $\mathcal{C}$. Second, we learn a privacy-preserving mapping $q_{\hat{C} \mid C}$ on the new alphabet, where $\hat{\mathcal{C}} = \mathcal{C}$. Third, the symbols in $\mathcal{B}$ are mapped to the representative examples $\hat{\mathcal{C}}$ based on the learned mapping $q_{\hat{C} \mid C}$. Our approach is summarized in Algorithm~\ref{alg:quantized PPM}  and Diagram~\ref{fig:diag}.


Our solution has several notable properties. 
To begin with, the privacy-preserving mapping $q_{\hat{C} \mid C}$ is learned on the reduced alphabet $\mathcal{C}$. Thus, we need to solve the convex optimization~\eqref{eq: convex opt} for only $|\mathcal{C}| |\hat{\mathcal{C}}|$ variables instead of $|\mathcal{B}| |\hat{\mathcal{B}}|$. In practice, $|\mathcal{C}| \ll |\mathcal{B}|$ and this results in major computational savings. 
Second, quantization and privacy-preserving optimization are done separately. Therefore, any quantization method can be easily combined with our approach. In particular, we can minimize the quantization error in the quantization step, and then our privacy mechanism guarantees the optimal mapping in terms of additional distortion.
Finally, quantization obviously yields a suboptimal privacy-accuracy tradeoff, since the quantization step is an additional source of distortion. However, in Theorem~\ref{thm:quantization error}, we quantify how quantization affects the privacy-accuracy tradeoff, and show that the levels of privacy that can be achieved are not affected, but come at the expense of a bounded amount of distortion.

%
%

In the rest of this section, we analyze Algorithm~\ref{alg:quantized PPM}. Algorithm~\ref{alg:quantized PPM} essentially solves  the following variant of problem \eqref{eq: convex opt}: 
\begin{align}
  \underset{p_{\hat{C} \mid C}}{\text{minimize}} \ \ & J(q_{A, C}, p_{\hat{C} \mid C})
  \label{eq:quantization problem} \\
  \text{subject to:} \ \ & \E{p_{C, \hat{C}}}{d(C, \hat{C})} \leq \Delta \nonumber \\
  & p_{\hat{C} \mid C} \in \rm{Simplex}; \nonumber
\end{align}
where alphabets $\mathcal{B}$ and $\hat{\mathcal{B}}$ are substituted for alphabets $\mathcal{C}$ and $\hat{\mathcal{C}}$, and the joint probability distribution over $A$ and $C$ is defined as 
\begin{align}
  q_{A, C}(a, c) = \sum_{b \sim c} p_{A, B}(a, b),
  \label{eq:quantized prior}
\end{align}
where $b \sim c$ means that the symbol $b$ is in the cluster represented by center $c$. The above equation aggregates the probability mass of all symbols in the cluster in its center. The symbols in $\mathcal{B}$ are mapped to $\hat{\mathcal{C}}$ according to
\begin{align}
  p_{\hat{C} \mid B}(\hat{c} \mid b) = q_{\hat{C} \mid C}(\hat{c} \mid \psi(b)),
  \label{eq:quantized mapping}
\end{align}
where $\psi: B \rightarrow C$ is a function that maps a symbol in $\mathcal{B}$ to a cluster center in $\mathcal{C}$. Note that the probability distributions that are associated with optimization \eqref{eq:quantization problem} are marked by $q$. 
We now prove our main claim.

\begin{theorem}
\label{thm:quantization error} Let $q_{\hat{C} \mid C}$ be a solution to problem \eqref{eq:quantization problem} and $p_{\hat{C} \mid B}$ be the corresponding mapping from $\mathcal{B}$ (Equation~\ref{eq:quantized mapping}). Moreover, let $\mathcal{C}$ be an alphabet such that $\max\limits_{b \in \mathcal{B}} \min\limits_{c \in \mathcal{C}} d(b, c) \leq r$. Then the privacy leakage $J(p_{A, B}, p_{\hat{C} \mid B})$ of the mapping $p_{\hat{C} \mid B}$ is equal to the value of the objective function of (\ref{eq:quantization problem}):
\begin{align*}
  J(p_{A, B}, p_{\hat{C} \mid B}) = J(q_{A, C}, q_{\hat{C} \mid C}),
\end{align*}
and its total distortion rate is no more than $r$ larger than the target~$\Delta$:
\begin{align*}
  \E{p_{B, \hat{C}}}{d(B, \hat{C})} \leq \Delta + r.
\end{align*}
\end{theorem}

\begin{proof}
The information-leakage equality can be proved as follows. First, both $J(p_{A, B}, q_{\hat{C} | B})$ and $J(q_{A, C}, q_{\hat{C} | C})$ can be rewritten as
\begin{align}
  J(p_{A, B}, q_{\hat{C} | B})
  \ = & \ H(p_A) + H(p_{\hat{C}}) - H(p_{A, \hat{C}}) \\
  J(q_{A, C}, q_{\hat{C} | C})
  \ = & \ H(q_A) + H(q_{\hat{C}}) - H(q_{A, \hat{C}}),
\end{align}
where
\begin{align}
  p(a, \hat{c})
  \ = & \ \sum_b q(\hat{c} | \psi(b)) p(a, b) \\
  q(a, \hat{c})
  \ = & \ \sum_c q(\hat{c} | c) q(a, c).
\end{align}
Second, note that
\begin{align}
  p(a, \hat{c})
  \ = & \ \sum_b q(\hat{c} | \psi(b)) p(a, b) \nonumber \\
  \ = & \ \sum_c q(\hat{c} | c) \sum_{b \sim c} p(a, b) \nonumber \\
  \ = & \ \sum_c q(\hat{c} | c) q(a, c) \nonumber \\
  \ = & \ q(a, \hat{c}).
\end{align}
So the two distributions are identical. Thus $H(p_{A, \hat{C}}) = H(q_{A, \hat{C}})$, and an analogous result holds for the entropies of the marginals. As a result, the privacy leakage of the mapping $q_{\hat{C} | B}$ on $\mathcal{B}$ is equal to the privacy leakage of the mapping $q_{\hat{C} | C}$ on $\mathcal{C}$.

The distortion inequality is proved as follows. First, note that \eqref{eq:quantized mapping} implies
\begin{align}
  q_{B, \hat{C}}(b, \hat{c})
  \ = & \ \sum_a q_{\hat{C} | B}(\hat{c} | b) p_{A, B}(a, b) \nonumber \\
  \ = & \ \sum_a q_{\hat{C} | C}(\hat{c} | \psi(b)) p_{A, B}(a, b).
\end{align}
Based on this equality, we can bound the distortion as
\begin{align}
  \E{q_{B, \hat{C}}}{d(B, \hat{C})}
  \ = & \ \sum_{b, \hat{c}} q(b, \hat{c}) d(b, \hat{c}) \nonumber \\
  \ = & \ \sum_{a, b, \hat{c}} q(\hat{c} | \psi(b)) p(a, b) d(b, \hat{c}) \nonumber \\
  \ = & \ \sum_{a, c, \hat{c}} q(\hat{c} | c) \sum_{b \sim c} p(a, b) d(b, \hat{c}) \nonumber \\
  \ \leq & \ \sum_{a, c, \hat{c}} q(\hat{c} | c) \sum_{b \sim c} p(a, b)
  [d(b, c) + d(c, \hat{c})] \nonumber \\
  \ = & \ \sum_{a, c, \hat{c}} q(\hat{c} | c)
  \underbrace{\sum_{b \sim c} p(a, b)}_{q(a, c)} d(c, \hat{c}) \ + \nonumber \\
  & \ \sum_{a, c} \underbrace{\sum_{\hat{c}} q(\hat{c} | c)}_1
  \sum_{b \sim c} p(a, b) d(b, \psi(b)) \nonumber \\
  \ \leq & \ \E{q_{C, \hat{C}}}{d(C, \hat{C})} +
  r \sum_{a, b} p(a, b) \nonumber \nonumber \\
  \ \leq & \ \Delta + r.
\end{align}
This concludes our proof. 

\end{proof}

\medskip

Theorem~\ref{thm:quantization error} states that the information leakage of the mapping $p_{\hat{C} \mid B}$ is the same as that of the optimized mapping $q_{\hat{C} \mid C}$. So we optimize the quantity of interest $J(p_{A, B}, p_{\hat{C} \mid B})$ in a time which is independent of the size of the input alphabet $\mathcal{B}$. The total distortion increases due to quantization, linearly with the maximum distance $r$ between any example $b$ and its closest representative example $\psi(b)$.

The maximum distance $r$ can be minimized by existing quantization techniques, e.g. online $k$-center clustering~\cite{charikar97incremental} and cover trees~\cite{beygelzimer06cover}. Both methods quantize data nearly optimally. In particular, if the minimum quantization error by $\abs{\mathcal{C}}$ examples is $r^\ast$, then the maximum error produced by these methods is $8 r^\ast$. Note that finding $\abs{\mathcal{C}}$ examples that minimize the quantization error is NP hard.

\begin{algorithm}[t]
  \caption{Quantized privacy preserving mapping.}
  \label{alg:quantized PPM}
  \begin{algorithmic}
    \State {\bf Input:} prior $p_{A, B}$
    \State
    \ForAll{$(a, c) \in (\mathcal{A}, \mathcal{C})$}
      \State $q_{A, C}(a, c) \gets \sum_{b \sim c} p_{A, B}(a, b)$
    \EndFor
    \State solve the convex optimization problem over $p_{\hat{C} | C}$:
\begin{align*}
  \underset{p_{\hat{C} | C}}{\text{minimize}} & \quad
  J(q_{A, C}, p_{\hat{C} | C}) \label{eq:quantization problem} \\
  \text{subject to} & \quad
  \E{p_{C, \hat{C}}}{d(C, \hat{C})} \leq \Delta \nonumber \\
  & \quad p_{\hat{C} | C} \in \rm{Simplex}; \nonumber
\end{align*}
\State return optimal solution $q_{\hat{C} | C}$
    \ForAll{$(b, \hat{c}) \in (\mathcal{B}, \hat{\mathcal{C}})$}
      \State $p_{\hat{C} | B}(\hat{c} | b) \gets q_{\hat{C} | C}(\hat{c} | \psi(c))$
    \EndFor
    \State
    \State {\bf Output:} mapping $p_{\hat{C} | B}$
  \end{algorithmic}
\end{algorithm}

\section{Datasets}\label{sec:dataSets}

In order to evaluate our framework, we apply it to three datasets. The first two datasets, the \emph{Census} data \cite{censusdata} and \emph{Iris} data \cite{irisdata} are well-known publicly available datasets. The third one, called \emph{Politics-and-TV}, is a dataset on political convictions and TV preferences, that we collected by conducting a survey, as explained in Section~\ref{sec:PoliticsTVdata}.

These three datasets were selected because each allows us to illustrate different components of our work. We use the \emph{Census} dataset to illustrate the basic performance of Algorithm~\ref{alg:PPM}. We evaluate Algorithm~\ref{alg:quantized PPM} on both the \emph{Iris} and \emph{Politics-and-TV} datasets. We start with the \emph{Iris} data because it is a simple low dimensional dataset that allows us to visualize the effect of our proposed distortion techniques. The \emph{Politics-and-TV} data has a high-dimensional alphabet and thus allows us to evaluate how quantization influences our ability to provide privacy. The \emph{Census} and \emph{Politics-and-TV} datasets have data that lie in discrete sets, while the \emph{Iris} dataset has continuous entries. We present the optimal  privacy-accuracy curve for each case, and give some insights on the privacy mappings.

\subsection{Census Dataset}

The \emph{Census} dataset is a well studied dataset in the Machine Learning community. Based on the 1994 Census, the dataset is a sample of the United States population, and contains both categoric and numerical features. More precisely, for each entry in the dataset, there are features such as age, workclass, education, gender, and native country, as well as income category (smaller or larger than 50k per year). For our purposes, we consider the information to be released publicly as the education, gender, and age, while the income category is the private information to be protected. It is noteworthy to know that about 76\% of the people in the dataset have an income smaller than 50k.

Our privacy mechanism in this case uses erasures. Erasure policies are ones in which we advise a user how to modify their public profile before it is released, by erasing 1, 2 or 3 pieces of information, in order to make it hard to infer income category.

The suggestion is tailored to each individual.

The joint probability distribution $p_{A,B}$ is estimated over the available data. Because of the discrete nature of the data, the low dimension of the feature space considered, and the large number of available observations (about 50,000 entries), the joint distribution can be estimated easily with very high confidence. In this case, there is essentially no prior mismatch.

\subsection{Iris Dataset}

The \emph{Iris} dataset has been used by the Machine Learning community extensively \cite{irisdata}. The dataset consists of four numerical attributes (petal length in $cm$, petal width in $cm$, sepal length in $cm$, and sepal width in $cm$) and one class attribute (Iris Setosa, Iris Versicolour, Iris Virginica) that identifies the particular category of Iris flower. There are 50 samples per class, for a total of 150 samples. It has been shown that using the 4 numerical attributes, it is possible to build very good classifiers to identify the type of Iris flower \cite{ucimlrepository}. We can visualize this by projecting the 4 numerical values on 2 principal components (see Fig.~\ref{fig:iris}), which explain 97\% of the variance. It is straightforward to see that the three flower classes are almost linearly separable in this space. We also see that the Iris Setosa has attributes that differentiate it much better, whereas there might be some confusion between the Iris Versicolour, and Iris Virginica. However, we again emphasize the fact that even though the latter two flowers are close in this space, the classifiers' accuracy is still very good (between 80\% and 100\% accuracy).

\psfrag{Mutual Information}[bc][cc][1][0]{{\scriptsize Mutual Information}}
\psfrag{Mutual information}[bc][cc][1][0]{{\scriptsize Mutual Information}}
\psfrag{l2 Distortion}{{\scriptsize $l_2$ distortion}}
\psfrag{Expected number of erasures}{{\scriptsize Expected number of erasures}}
\psfrag{Number of people}[c][c]{{\scriptsize Number of people}}
\psfrag{Demographics}{}
\psfrag{Gender}[c][c]{{\scriptsize Gender}}
\psfrag{Age}[c][c]{{\scriptsize Age}}
\psfrag{Politics}[c][c]{{\scriptsize Politics}}
\psfrag{Male}[c][c]{{\scriptsize \textcolor{white}{Male}}}
\psfrag{Female}[c][c]{{\scriptsize \textcolor{white}{Female}}}
\psfrag{Democrat}[c][c]{{\scriptsize \textcolor{white}{Democrat}}}
\psfrag{Republican}[c][c]{{\scriptsize \textcolor{white}{Republican}}}
\psfrag{65+}[c][c]{{\tiny 65+}}
\psfrag{51-65}[c][c]{{\tiny \textcolor{black}{ 51-65}}}
\psfrag{31-50}[c][c]{{\tiny \textcolor{white}{ 31-50}}}
\psfrag{21-30}[c][c]{{\tiny \textcolor{white}{ 21-30}}}
\psfrag{13-20}[c][c]{{\tiny \textcolor{white}{ 13-20}}}
\psfrag{<13}[c][c]{{\tiny \textcolor{white}{ $<13$}}}

\begin{figure*}
\minipage{0.32\textwidth}
\psfrag{Male}{{\footnotesize Male}}
\includegraphics[width=\linewidth]{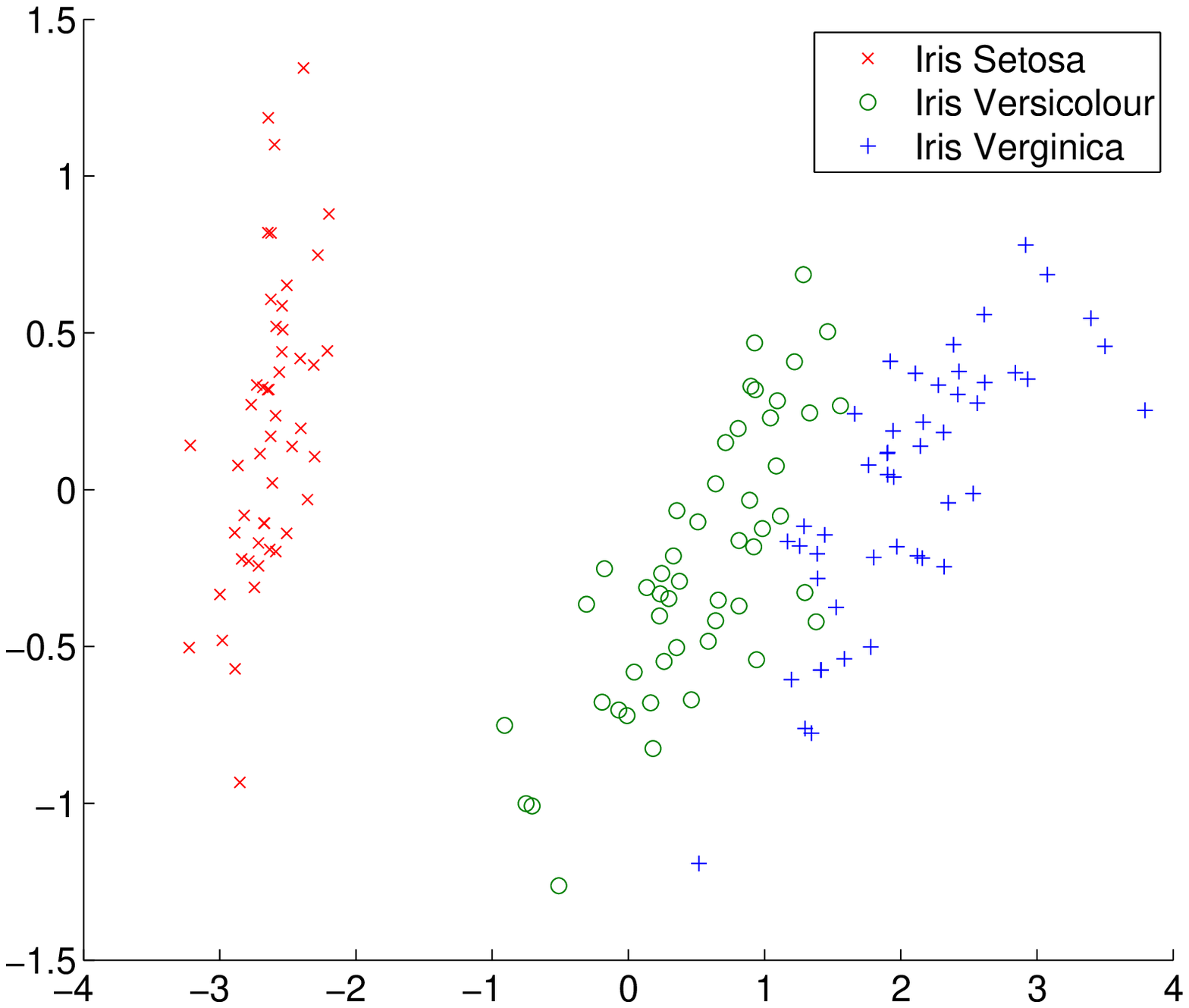}
\caption{Projection of the 4 numerical attributes in the 2 first principal components}
\label{fig:iris}
\endminipage\hfill
\minipage{0.32\textwidth}
\includegraphics[width=\linewidth]{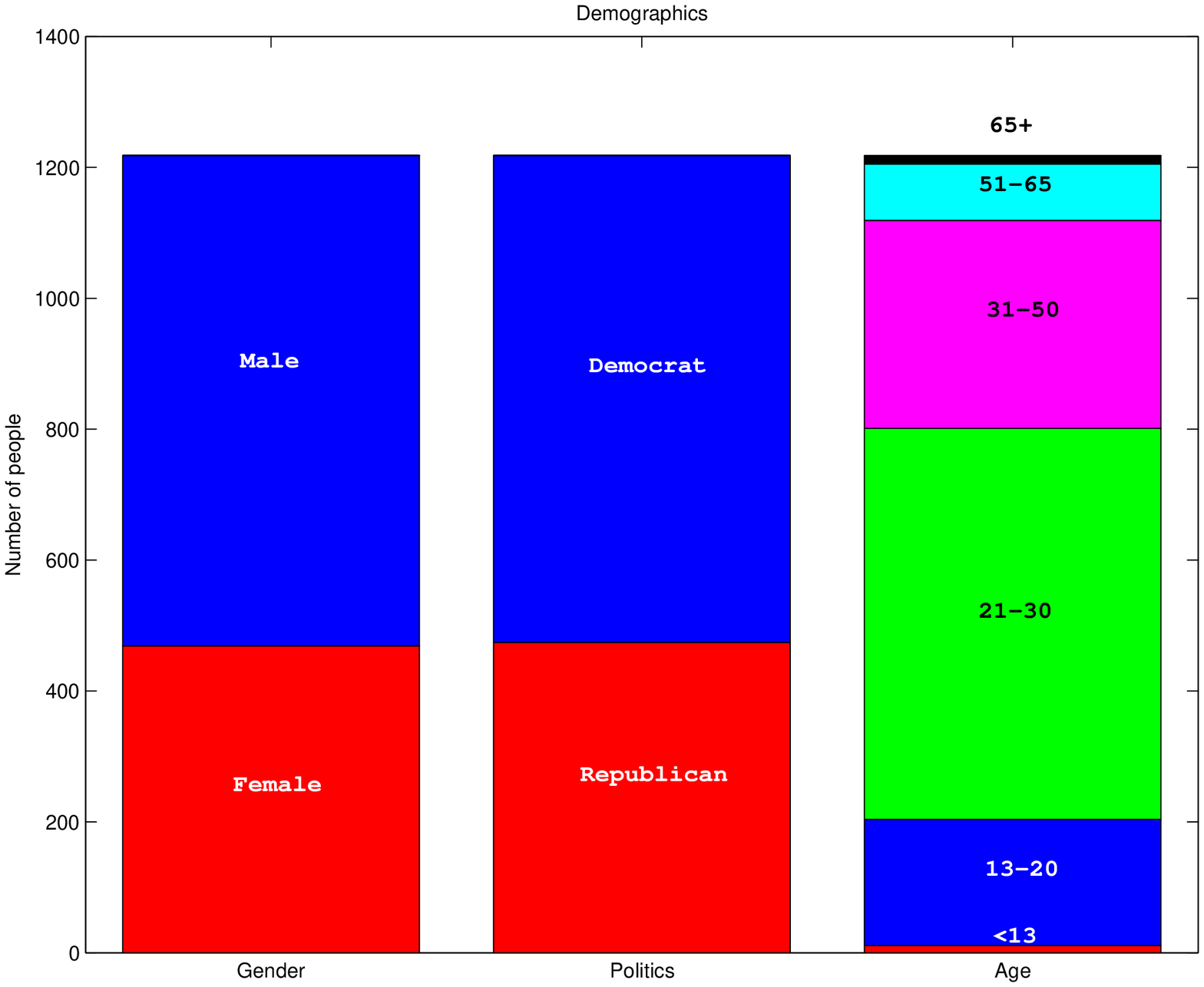}
\caption{Demographics of Survey Population}
\label{fig:PoliticsTV-Demographics}
\endminipage\hfill
\minipage{0.32\textwidth}
\includegraphics[width=\linewidth]{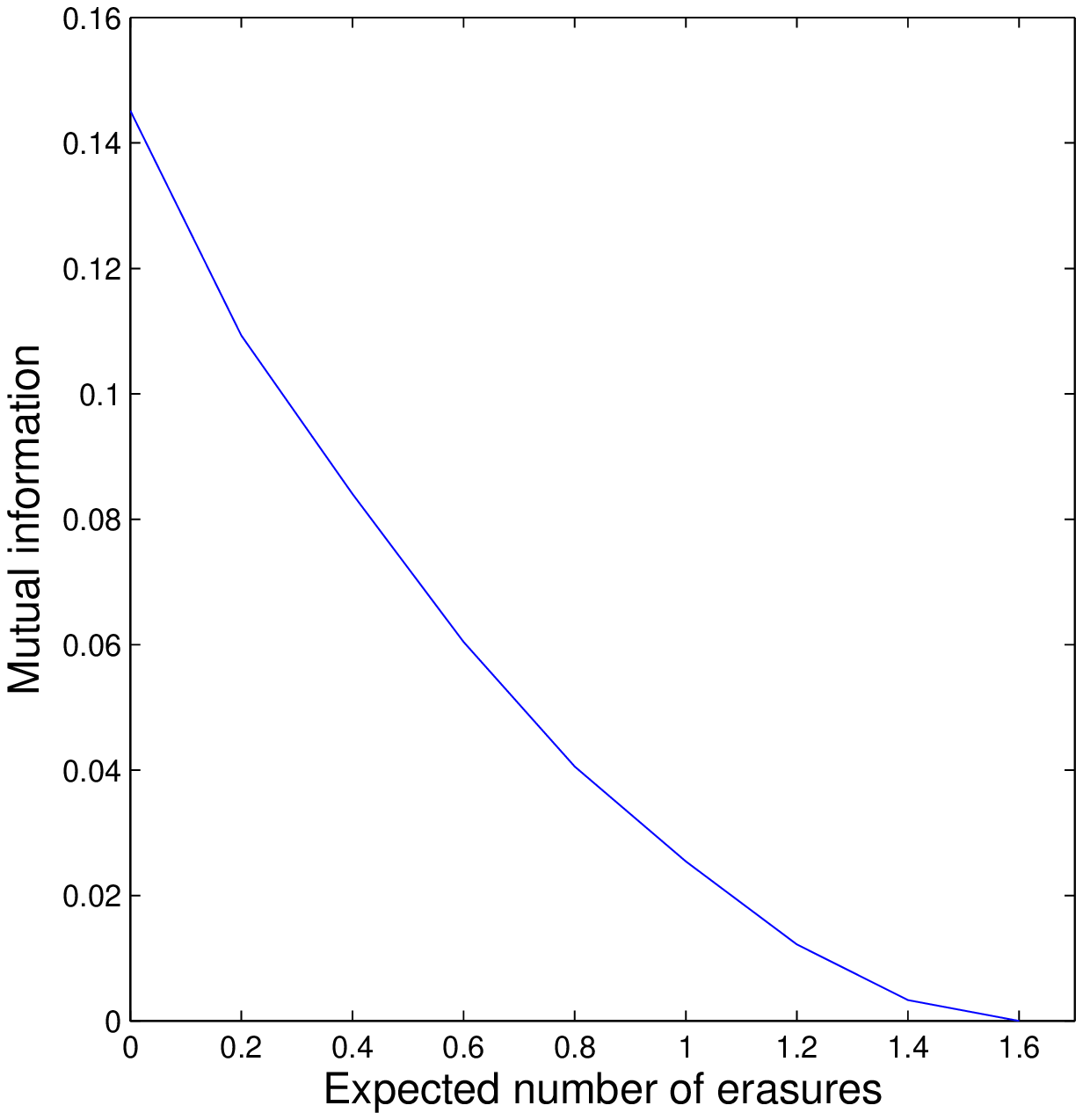}
\caption{Census data: Privacy Distortion curve}
\label{fig:censusPrivacy}
\endminipage
\end{figure*}

Because of the continuous nature of the data, we consider kernel based tools to estimate the joint probability distribution. Our prior distribution may hence be slightly mismatched compared to the true prior.

\subsection{Politics and Media Dataset}\label{sec:PoliticsTVdata}

The \emph{Politics-and-TV} dataset gathers data on political convictions and TV preferences of viewers in the USA in Fall 2012. The collection of such data was motivated by large scale surveys such as \cite{buzztvsurvey,Experian-survey}, which illustrated that the audiences for a number of TV shows can be distinctly characterized. Opinion polls have also published articles in the press with lists of top-10 or 20 TV shows that are most indicative of political affiliation. For example, \emph{The Colbert Report} is predominantly watched by Democrats, whereas \emph{Fox News} and \emph{Swamp Loggers} are primarily watched by Republicans. We thus started from the premise that it is possible to use public information about a user's TV preferences, such as the list and ratings of TV shows he watches, to infer some private information, namely political convictions. It should be noted that fewer than 1\% of Facebook users disclose their political views in their public profile, which seems to indicate that political convictions are deemed private information. We describe hereafter the data collection process, and our dataset.

\noindent\textbf{Data Collection}: We designed a survey that users take voluntarily. In our survey, users were first asked to provide demographic information (gender, age group, state they live in) as well as their political convictions (Democrat, Republican). Then users were asked to complete a sequence of 6 panels, each panel presenting the user with 6-8 TV shows of a certain genre, namely Sitcoms, Reality Shows, TV series, Talk Shows, News, and Sports, for a total of 50 TV shows. Users were asked to rate only those TV shows that they watched on a scale from 1 to 5--- the usual star rating system. After providing their ratings, users were shown, for each genre, how their ratings compared with the average ratings given by Democrats and Republicans. In our privacy policy, users were informed that no private information that can be used to identify an individual was stored---  we did not store cookies, nor IP addresses, etc. Thus the data collected is by consenting users.

We ran our survey in two phases. In phase 1 (October 2012), we ran it on Mechanical Turk requesting only US-based workers. An initial experiment revealed that 80\% of users completing the survey were Democrats. To diminish this bias, we reran the survey in two batches. For the first batch, we limited the  user pool to Democrats only, and in the second batch we limited it to Republicans only. This mechanism helped although it still did not produce equal numbers of Democrats and Republicans. In total, we obtained 854 surveys, with 518 Democrats and 336 Republicans. In phase 2 (November 2012), we launched our survey on the public web at www.PoliticsandMedia.org. We drove traffic to the survey website by running   advertising campaigns on MyLikes.com and Google AdWords, shortly before the U.S. 2012 presidential election. From this, we obtained another 364 completed surveys, with 226 Democrats and 138 Republicans. We conducted this survey in two places (Mechanical Turk and the Web) to create more diversity of users in our survey. An advantage of the Mechanical Turk approach is that users are incentivized to properly complete the survey. We threw out surveys which were clearly never finished, e.g. no ratings,  and the numbers above reflect the final retained surveys.

\noindent\textbf{Dataset}: The  dataset contains entries for 1,218 users, broken into 744 Democrats, and 474 Republicans. For each user, the dataset entry is a vector $[\mathrm{age}, \mathrm{gender}, \mathrm{state}, \mathrm{politics}, r_1, \ldots... r_{50}]$ where $r_i\in\{0,1,\ldots,5\}$ is the user's star rating for show $i$ if the user rated the show, and 0 otherwise.
The 5 most watched TV shows are \emph{The Daily Show with Jon Stewart, The Colbert Report, NFL, The Big Bang Theory,} and \emph{Family Guy}. Figure~\ref{fig:PoliticsTV-Demographics} shows the demographics of the 1,218 users in the dataset. 

In the sequel, we will consider two versions of the rating vector: the 5-star rating vector $R\in \{0,1,\ldots,5\}^{50}$, and the binarized rating vector $B\in \{0,1\}^{50}$. The binarized rating $b_i$ of show $i$ is obtained by setting $b_i= 1$ if the original rating $r_i>=4$ clearly indicating that the user likes the show, and $b_i=0$ otherwise.

\section{Results}\label{sec:Results}

\subsection{Baseline Convex optimization on Census Dataset}


We demonstrate here a direct application of the convex optimization approach Algorithm-1 described earlier on the Census dataset. This can be seen as a simple application as we do not need to apply a quantization step. For this dataset, we will use the \emph{erasure-distortion} approach meaning that our proposed distortion to an individual's public data (age, education, and gender) may be to remove a subset of features. In this way, we distort without lying, and our distortion metric is the number of erasures.


Formally, let $B(u)=(b_1,b_2,b_3, a)$ be the features of  user $u$, where $b_1 \in \{male,female\}$, $b_2 \in \{$\emph{young, adult, old}$\}$ and $b_3 \in \{$\emph{high-school, college degree, master degree, doctorate}$\}$. The feature $a$ is the private attribute defined as  $a \in \{high,low\}$ where high/low refers to an income above/below 50K\$ respectively. In this case the output alphabet $ \mathcal{\hat{B}}$ after the privacy mapping is larger than the input alphabet $\mathcal{B}$ as each feature can be replaced by an erasure. Because of the mapping restriction $p_{\hat{b}|b}$ can have non zero values if $b$ and $\hat{b}$ differ only in positions where $\hat{b}$ has an erasure. We define the distortion metric $d(\hat{b},b)$ as the number of erasures in $\hat{b}$, when $b$ and $\hat{b}$ match in non-erasure positions and $d(\hat{b},b)=\infty$ otherwise. \\

We have tested the algorithm for different distortion constraint values and obtained the privacy-distortion curve shown in Fig. \ref{fig:censusPrivacy}. The y-axis captures the privacy leakage measured by the mutual information. The x-axis quantifies the distortion in terms of average number of erasures. Without any of our distortions (0 erasures), the privacy leakage, or mutual information, is 0.142 bits. If, on average, we erase one of the three features in these user profiles, then the privacy leakage drops to roughly 0.025 bits. This can be interpreted as requiring an adversary to ask many more questions in order to learn the private information. Perfect privacy (mutual information is zero) is obtained when the expected erasures is $1.5$ features (out of three). This confirms that gender, age and education are related to one's income.

Since the privacy-distortion curve alone does not provide  much insight on the privacy mapping, we have represented some specific case of mappings in Table \ref{tab:census}. It is interesting to see that some different categories get mapped to an identical vector, for example, row 2 and row 5 are both mapped to `male' with 2 erasures. This illustrates the confusion created by our distortions; an adversary that sees such an output cannot determine its original form, and will likely learn next to nothing about these individuals incomes.
\begin{figure*}
\centering
\begin{tabular}{| c | c | c | c || c | c | c |}
\cline{2-7}
\multicolumn{1}{c}{} & \multicolumn{3}{| c |}{Original features}&\multicolumn{3}{|c|}{Private mapping}  \\
\hline

$<50$k&male & young & College degree & male & - & - \\ \hline
N&male & adult & College degree & - & adult & College degree \\ \hline
$<50$k&male & young & High School & male & - & High school \\ \hline
N&male & adult & High School & male & - & High School \\ \hline
$<50$k&female & young & High School & - & - & - \\ \hline
$<50$k&female & young & College degree & - & - & - \\ \hline
$>50$k&male & adult & Masters degree & male & - & - \\ \hline
$<50$k&female & adult & College degree & - & adult & College degree \\ \hline
\end{tabular}
\caption{Most probable mapping for the Top 8 categories in the Census Dataset. Initially some set of attributes may be highly correlated with income (denoted by $<50k$ and $>50k$), or be more neutral (denoted by N).}
\label{tab:census}
\end{figure*}

\subsection{Mismatched prior and quantization on Iris data}

\begin{figure*}
\minipage{0.32\textwidth}
\includegraphics[width=\textwidth]{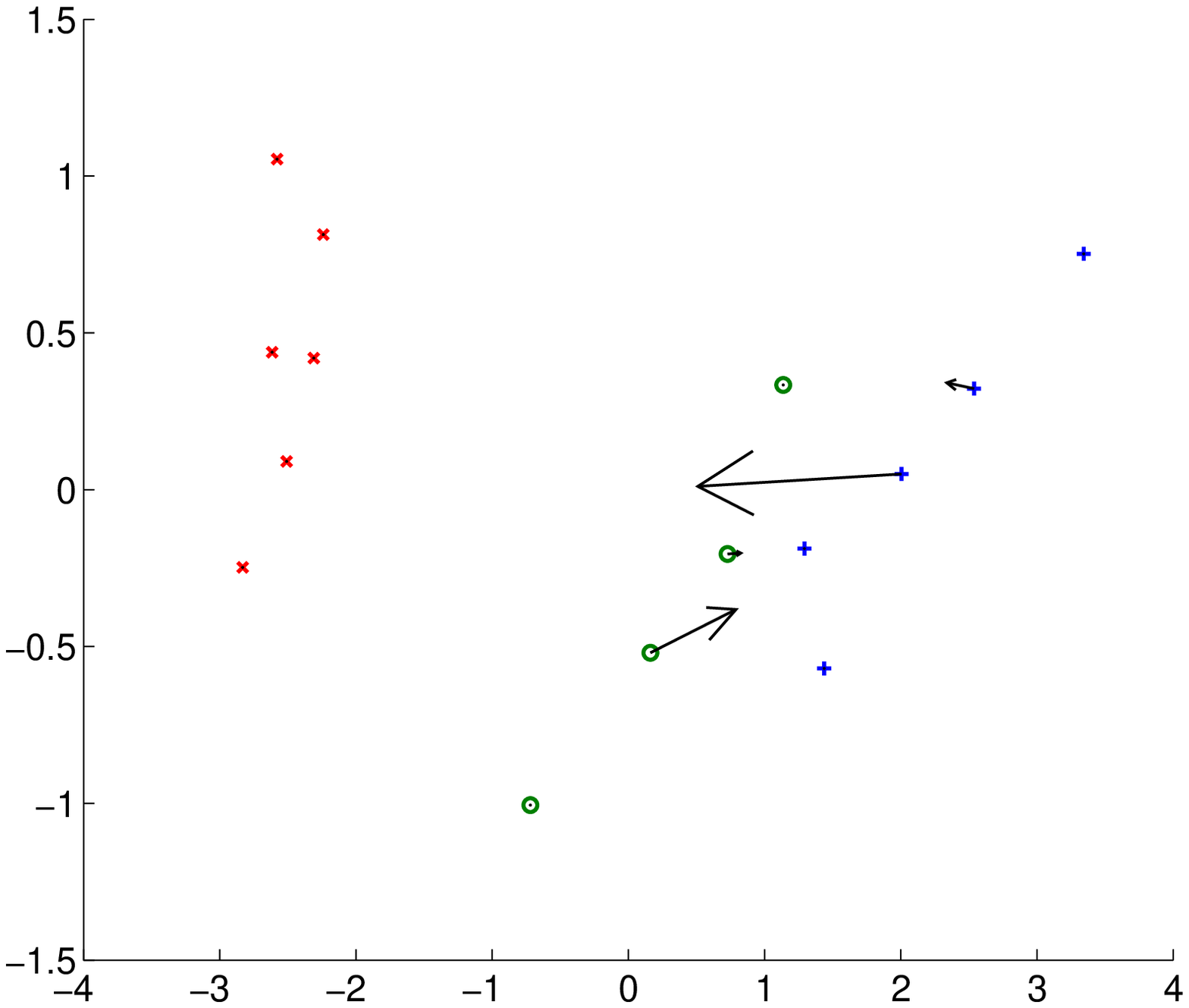}
\caption{Iris data: Each point represents a cluster. The arrows show graphically the privacy mapping for a low distortion of $0.5$ in $l_2$ distance. Clusters on the far left are untouched}
\label{fig:velocity_1}
\endminipage\hfill
\minipage{0.32\textwidth}
\includegraphics[width=\textwidth]{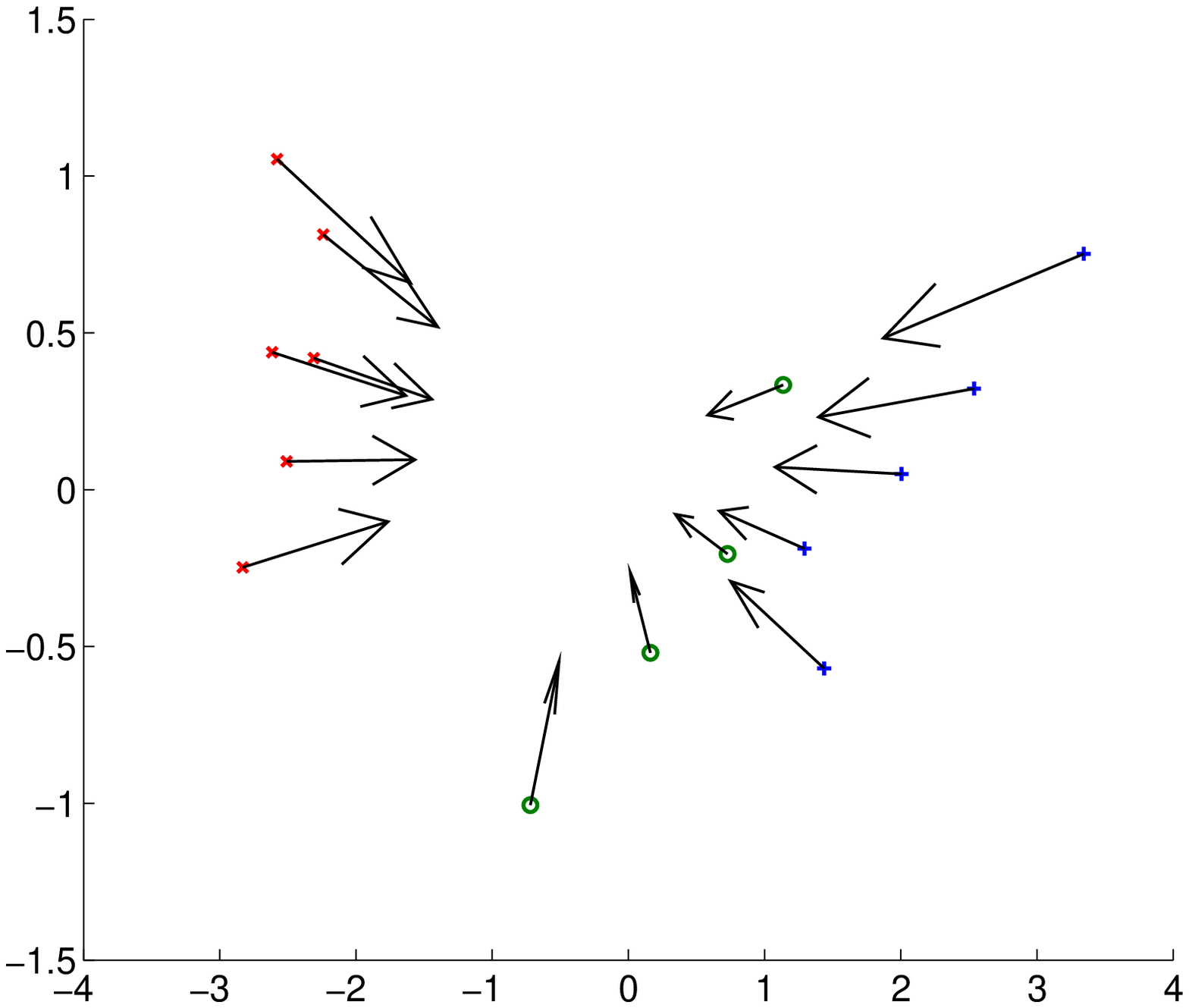}
\caption{Iris data: When we allow high distortion however, all clusters are affected, even those on the far left.}\label{fig:velocity_2}
\endminipage\hfill
\minipage{0.32\textwidth}
\includegraphics[width=\textwidth]{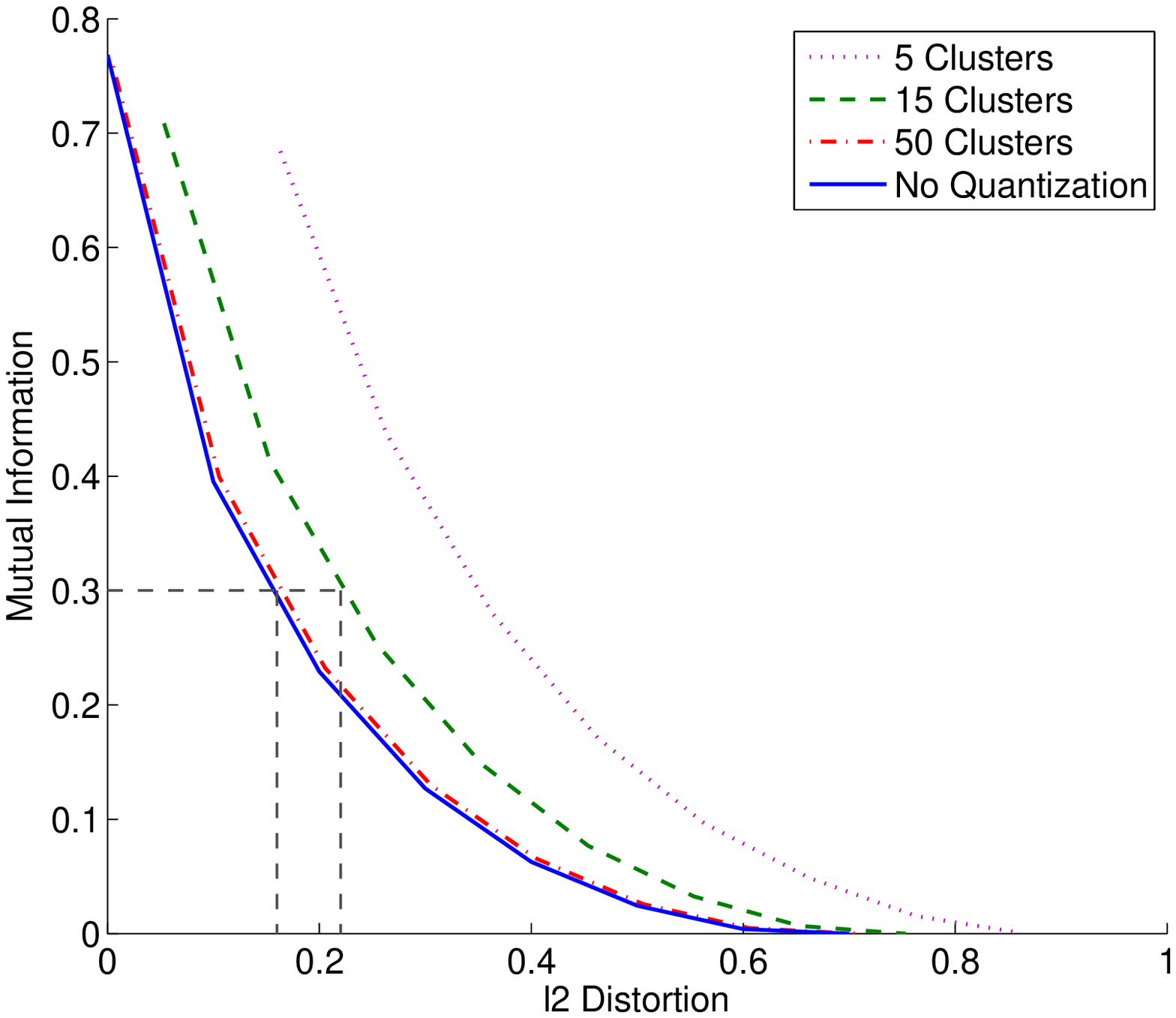}
\caption{Iris data: Privacy Distortion curves, one for each level of quantization.}
\label{fig:quantization}
\endminipage
\end{figure*}

Recall that the Iris dataset has a slightly mismatched prior and the privacy goal is to make it hard to classify the Iris Virginica flower correctly. Intuitively we can do this by blurring the distinction between the Virginica and Versicolour flowers; yet this  needs to be done without diminishing the ability to correctly classify the Setosa flower. This simple example will allow us to directly compare the effect of quantization by running both Algorithm-1 (without quantization) and Algorithm-2 (with quantization) on this data.

Let the private attribute be $a \in \{$\emph{not Virginica, Virginica}$\}$, and our observed behavior $B$ be the petal and sepal, length and width attributes.  Because these features are continuous, we estimate the probability density $p(A,B)$ using a gaussian kernel estimator, with bandwidth fitted through cross-validation, for each set of flower. We have therefore an estimate of the conditional densities $f_{B|A}$ for each $A$. We further sample the distribution to derive a discrete approximation of the joint density  $p_{A,B}$ that is needed for our convex optimization problem. In this case, we use the $\mathcal{L}_2$-norm as the distortion metric. This is a natural distortion metric for any optimizations seeking to minimize a squared error.

We tried different quantization granularities, obtained by applying the traditional $k$-means clustering method and changing the number of clusters. The privacy leakage versus distortion tradeoff is shown Fig.~\ref{fig:quantization}.  The curve labeled ``No Quantization" was obtained using Algorithm-1, whereas the others were generated using Algorithm-2 for different values of $k$ the number of clusters. We see that using 50 clusters is nearly identical to the behavior without quantization. This is a first confirmation that our quantization approach is sound. Furthermore, even quantization with 15 clusters achieves a privacy-distortion tradeoff that is quite similar to the tradeoff incurred when no quantization is used. This is very encouraging as it indicates that quantization does not penalize the mapping in any substantive way. For example, consider the mutual information of 0.03. For this level of privacy, we introduce an amount of distortion equal to 0.18. However if we introduce clustering to improve on computation complexity, then we need a somewhat larger distortion of 0.21 to achieve the same level of privacy. This is a small penalty since clustering reduces the number of input variables used in the privacy-accuracy optimization which, in turn, reduces the overall complexity of determining the optimal privacy mapping. This illustrates how clustering can make the privacy-distortion optimization significantly more tractable without incurring a large penalty in terms of distortion.

We show the privacy mapping in Fig.~\ref{fig:clusters}. Each point on these curves represents a cluster, and the arrow illustrates the suggested new value (i.e. distortion) that our  mapping determines. For low levels of distortion (top plot) the clusters on the left, containing the Iris Setosa flowers, are barely distorted. For the higher distortion level, all of the clusters are distorted. We can see that the distinction between the blue and green clusters has grown more blurred, while the Setosa flower clusters on the left still remain clearly separable from the others, allowing their correct classification.  This validates the fact that while we distort the data in order to provide privacy, it can still be used for some (approved) inference purposes.

\subsection{Mismatched prior and quantization on Politics-and-TV Data}
We demonstrate here a more realistic privacy preservation application over the Politics-and-TV dataset described earlier.

\psfrag{Political Distribution}[cc][cc][1][0]{{\scriptsize Ratings}}
\psfrag{Political Distribution for Shows Rated}[cc][cc][1][0]{{\scriptsize }}
\psfrag{D}[cc][cc][1][-90]{{\tiny D}}
\psfrag{R}[cc][cc][1][-90]{{\tiny R}}

\psfrag{A}[cr][cr][1][-60]{{\scriptsize Modern Family}}
\psfrag{B}[cr][cr][1][-60]{{\scriptsize Community}}


\psfrag{C}[cr][cr][1][-60]{{\scriptsize Glee}}
\psfrag{W}[cr][cr][1][-60]{{\scriptsize The Mentalist}}

\psfrag{E}[cr][cr][1][-60]{{\scriptsize The Daily Show}}
\psfrag{F}[cr][cr][1][-60]{{\scriptsize O'Reilly Factor}}

\psfrag{G}[cr][cr][1][-60]{{\scriptsize FOX News}}
\psfrag{H}[cr][cr][1][-60]{{\scriptsize NBC News}}

\psfrag{I}[cr][cr][1][-60]{{\scriptsize NBA}}
\psfrag{J}[cr][cr][1][-60]{{\scriptsize Nascar}}

\psfrag{Expected hamming distortion per rating}[tc][cc][1][0]{{\scriptsize Expected Hamming distortion per binarized rating}}
\psfrag{Expected $l_2$ distortion per rating}[tc][cc][1][0]{{\scriptsize Expected $l_2$ distortion per rating}}

\psfrag{TPR}[bc][tc][1][0]{{\scriptsize TPR}}
\psfrag{FPR}[tc][bc][1][0]{{\scriptsize FPR}}
\psfrag{Log Regression ROC}[cc][cc][1][0]{{\scriptsize }}
\psfrag{with ratings}[cl][cl][1][0]{{\scriptsize actual ratings}}
\psfrag{binary ratings}[cl][cl][1][0]{{\scriptsize binarized ratings}}
\psfrag{perturbed ratings, delta = 1}[cl][cl][1][0]{{\scriptsize distorted ratings $\Delta=1$}}
\psfrag{perturbed ratings, delta = 2}[cl][cl][1][0]{{\scriptsize distorted ratings $\Delta=2$}}

\begin{figure}[t]
\includegraphics[width = \linewidth]{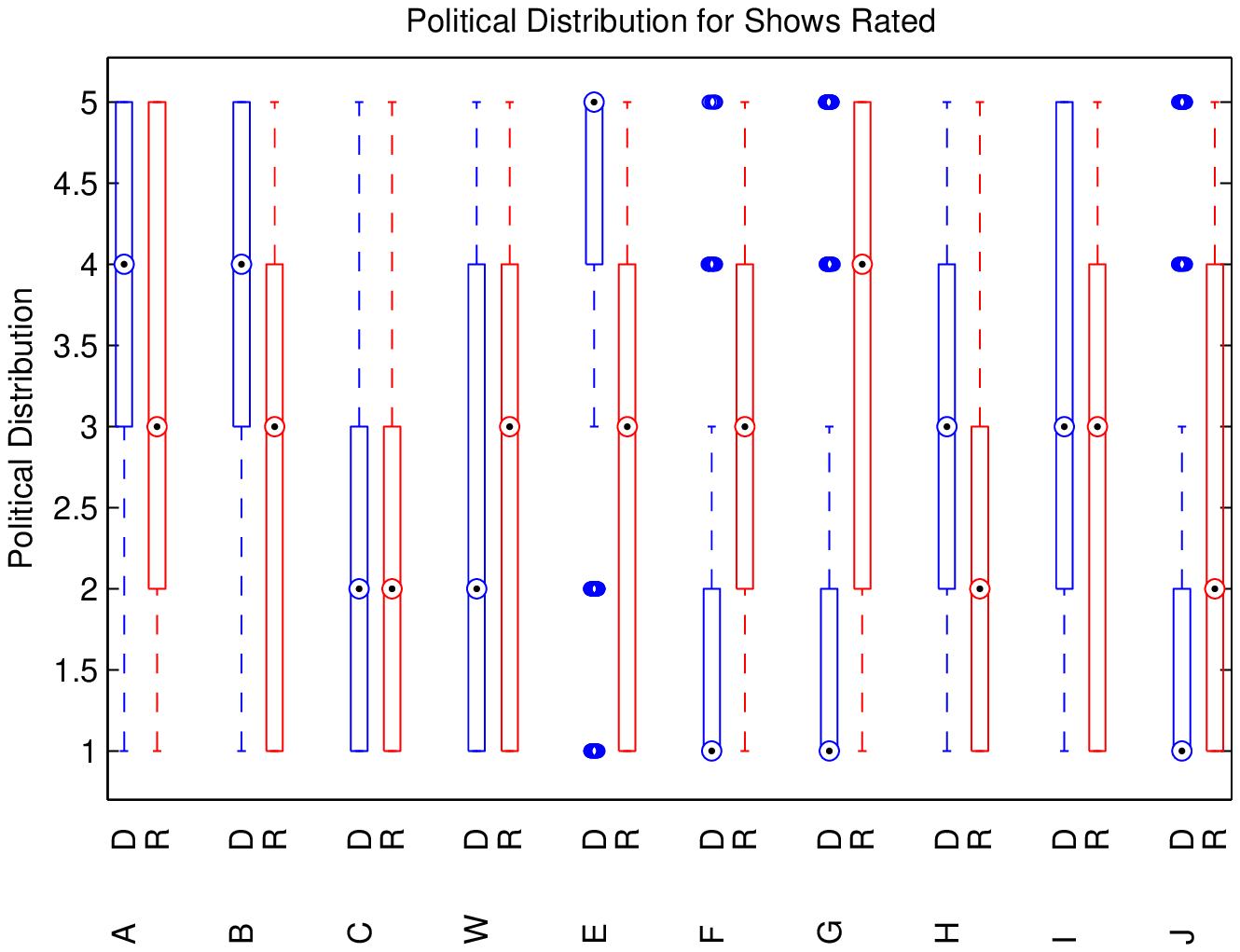}
\caption{Box plots of ratings for 12 TV shows by Democrats (D) and Republicans (R)}
\label{fig:politicsBox}
\end{figure}

\psfrag{Mutual Information}[bc][cc][1][0]{{\scriptsize Mutual Information}}
\psfrag{Expected hamming distortion per rating}[tc][cc][1][0]{{\scriptsize Expected Hamming distortion per rating}}
\psfrag{Expected L2 distortion per rating}[tc][cc][1][0]{{\scriptsize Expected L2 distortion per rating}}
\psfrag{TPR}[bc][tc][1][0]{{\scriptsize TPR}}
\psfrag{FPR}[tc][bc][1][0]{{\scriptsize FPR}}
\psfrag{Log Regression ROC}[cc][cc][1][0]{{\scriptsize }}

\begin{figure*}[t!]
\minipage{0.32\textwidth}
\includegraphics[width=\linewidth]{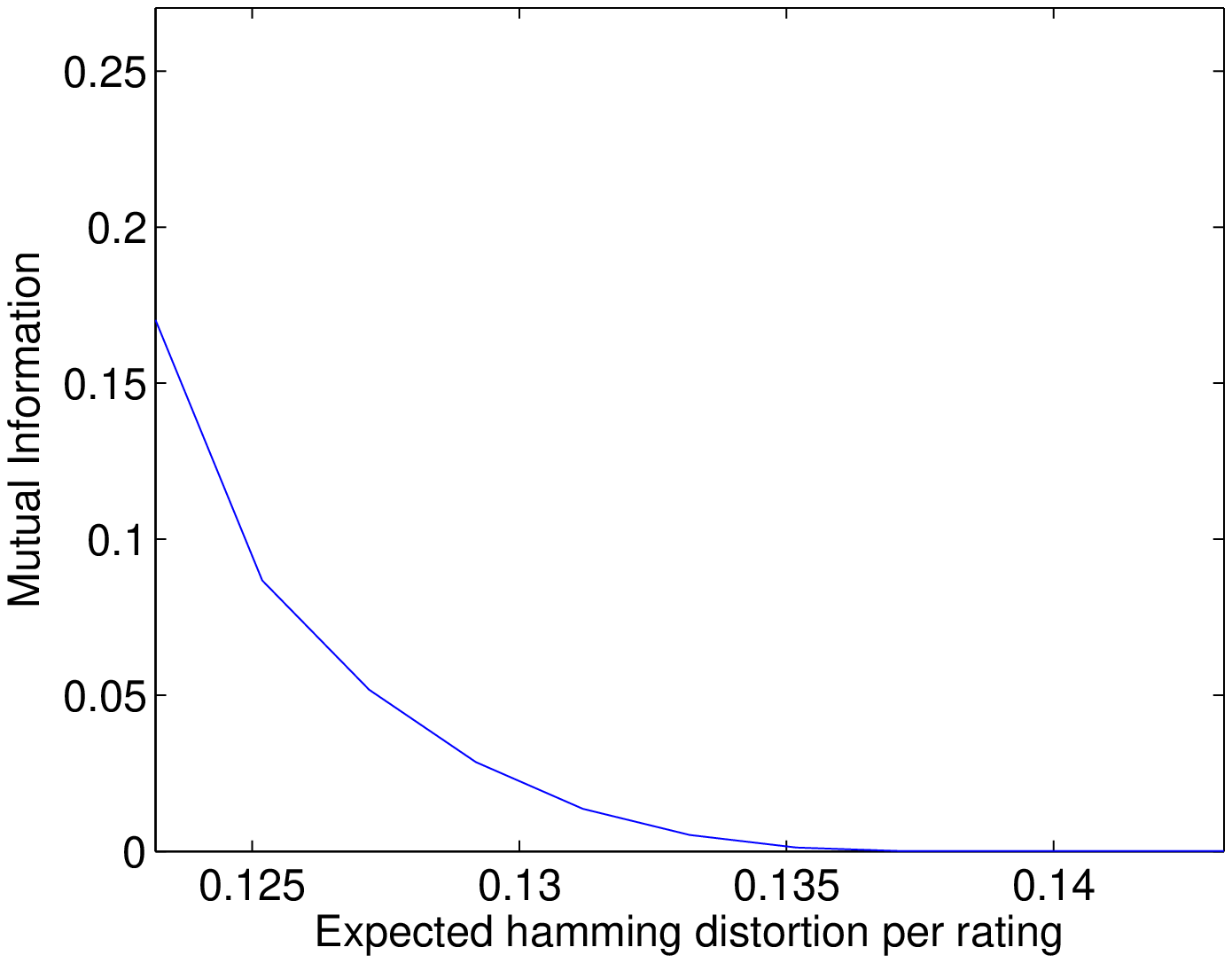}
\caption{Politics \& TV data: Privacy-accuracy trade-off on binarized ratings after quantization. Quantization introduces most of the distortion.}
\label{fig:privacyPolitics}
\endminipage\hfill
\minipage{0.32\textwidth}
\includegraphics[width=\linewidth]{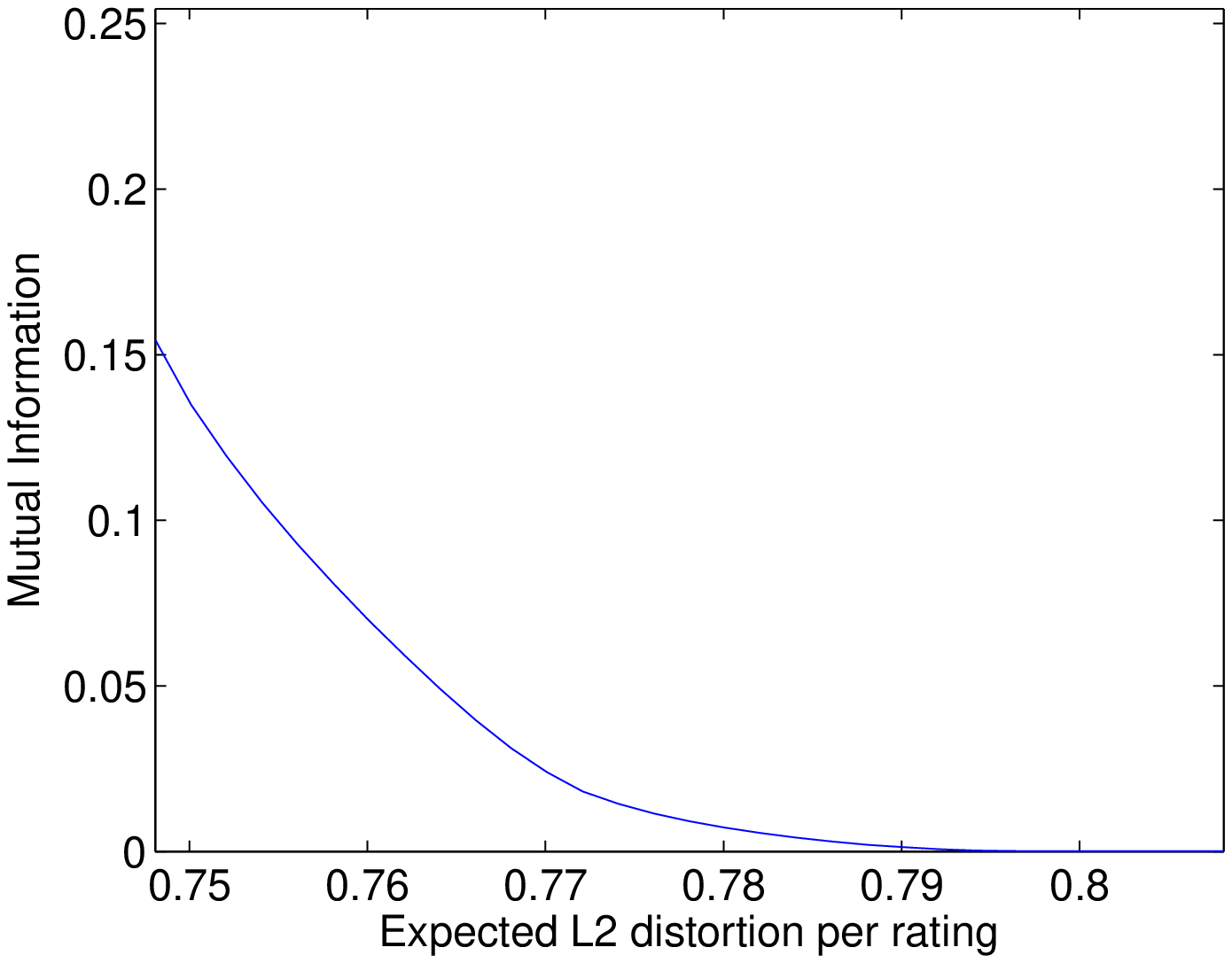}
\caption{Politics \& TV data: Privacy-accuracy trade-off on actual ratings after quantization.}
\label{fig:dist_MI}
\endminipage\hfill
\minipage{0.32\textwidth}
\includegraphics[width=\linewidth]{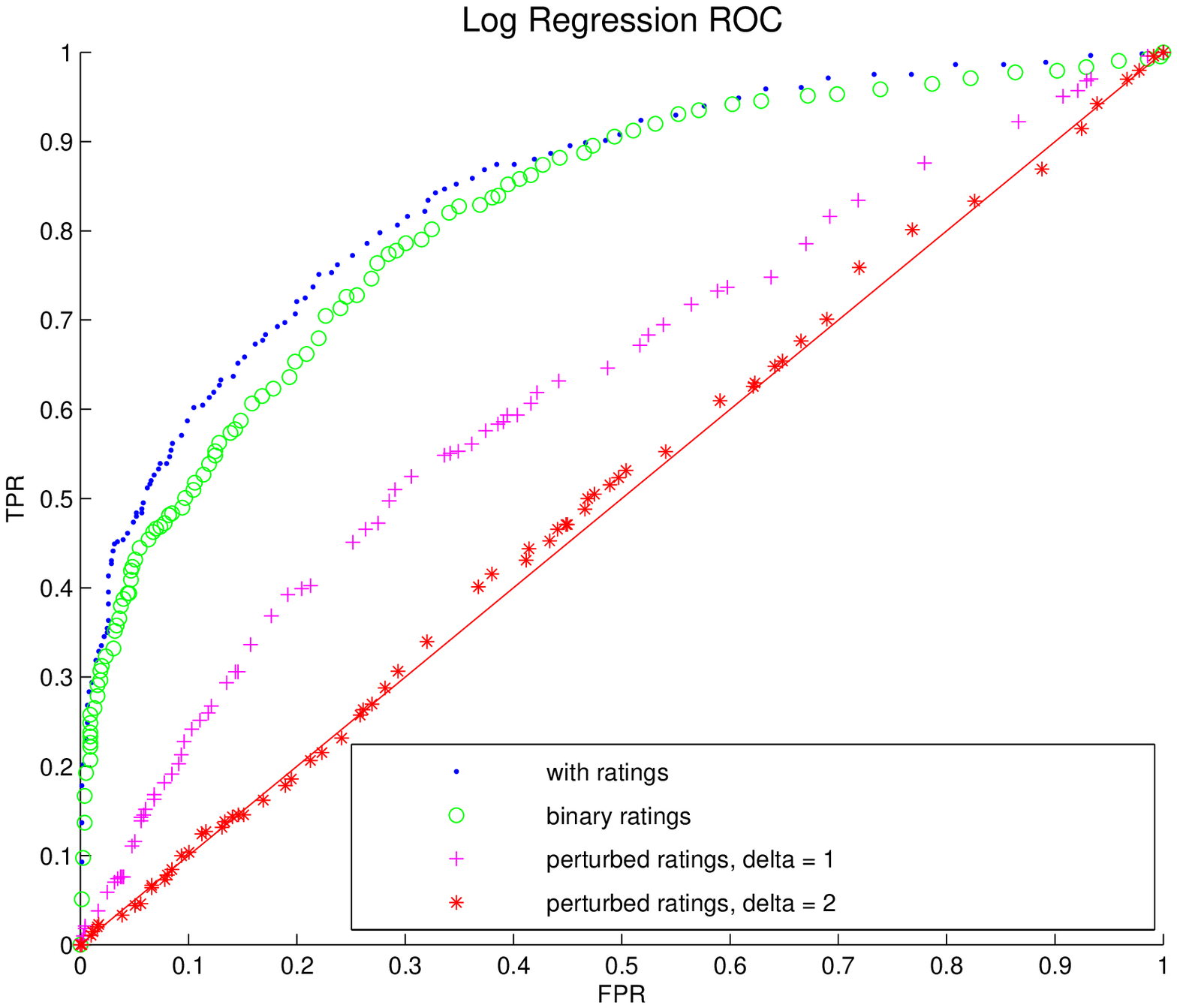}
\caption{Politics \& TV data: ROC curve of a logistic regression classifier for the political affiliation based on TV show ratings}
\label{fig:logregROC}
\endminipage
\end{figure*}

Consider the setting where a user wishes to release his TV show ratings $R\in\{0,1,\ldots,5\}^{50}$ (or $B\in\{0,1\}^{50}$), in the hope of getting good recommendations, but is concerned about them leaking information about his political affiliation $A\in\{\mathrm{Democrat},\mathrm{Republican}\}$. Note that although we focus on the case where the private data is a single variable representing political affiliation, the privacy-accuracy framework \cite{alerton2012privacy} can handle protecting a set of private variables, e.g. we could protect any subset of a user's three attributes [age,gender, politics]. The rating vector $R$ (reps. $B$) lives in a large alphabet of size $6^{50}$ (resp. $2^{50}$)\footnote{The number of survey samples is small relative to the size of the alphabet, and estimating the prior $p_{A,R}$ from the dataset may lead to a mismatched prior. We address the issue of the mismatched prior in \cite{Salamatian-Arxiv2013}.}. Solving (\ref{eq:convex opt}) over $6^{100}$ variables would be untractable, and justifies resorting to quantization. In this section, we first describe the privacy threat on political affiliation from the release of TV show ratings, then we characterize the privacy-accuracy trade-off under quantization. We illustrate the success of our privacy approach by showing how an inference algorithm degrades down to an uninformed guess at perfect privacy. Finally, we compare the quality of recommendations based on the actual user ratings versus the privatized ratings.

\noindent\textbf{Privacy threat}: 
The threat comes from the underlying existence of TV shows that are highly correlated with political affiliation, e.g. \emph{The Daily Show} is predominantly liked by Democrats, while \emph{Fox News} is preferred by Republicans. Fig. \ref{fig:politicsBox} shows boxplots of ratings for 12 shows---two shows from each genre in the dataset--- by Democrats and Republicans. 

Those shows for which there is little overlap in the opinions of Republicans and Democrats clearly demonstrate high correlation between political affiliation and opinion of those shows. Such shows have high discriminative power that inference algorithms can exploit. There exists a broad variety of shows in terms of their discriminative power - some are very much so, while others exhibit low correlation. Users who rate highly shows such as \emph{The O'Reilly Factor}, or \emph{The Daily Show}, may be facing a stronger threat than those who only watch and rate shows with little discriminating power. Broadly speaking, across our 50 shows, we found that roughly one third of them have strong correlation with political affiliation. 

In order to understand the threat inherent in this dataset, we quantify the potential privacy leakage using mutual information $I(A;R)$. 

To provide an illustrative example, we thus consider a reduced set of our data for which we can compute the mutual information. We consider the top 5 most seen TV shows, and use the binarized version of the rating vector with ratings in $\{0,1\}^{5}$. For this case, we observe that the mutual information between the observed features and the political orientation is already at $0.191$ bits. An adversary, with this information on hand (the 5 tuple of binarized ratings), could use a maximum a posteriori (MAP) detector and guess the political affiliation of somebody with an accuracy of $71\%$. Hence, the privacy threat is real. Note that because mutual information is a non-decreasing function, as we add additional shows, the threat  either stays the same or increases.

\noindent\textbf{Privacy-accuracy trade-off}: We now apply our quantization  approach and investigate its impact on the privacy-distortion tradeoff. We first consider the full dataset (all 50 shows) with the binarized version of the ratings. For this scenario we use an \emph{exchange-distortion} in which we exchange on TV show for another. We use Algorithm-1 that first quantizes the data using a clustering algorithm (k-means with a Hamming distance metric) into 25 clusters; then we apply the convex optimization on the quantized points.  The resulting trade-off curve is depicted in Fig. \ref{fig:privacyPolitics}. The curve shows that the quantization step alone introduces an average Hamming distortion of about $12\%$ (leftmost point on x-axis) per rating, or $6.1$ over all 50 shows, and results in a mutual information of $0.189$ on the representative points (cluster centers). As this is still high, we are motivated to apply further distortion. Fig. \ref{fig:privacyPolitics} shows that using the optimal privacy preserving scheme resulting from convex optimization, we can steadily decrease the privacy threat with increasing distortion. Not only is our privacy-distortion curve properly behaved, but small increases in Hamming distance bring the privacy leakage down quickly. Moreover, we can achieve perfect privacy ($I=0$) at the cost of an additional $3\%$ in average Hamming distortion (beyond the clustering distortion).  Perfect privacy is achieved at an overall Hamming Distortion of less than 7 out of 50; put alternatively, perfect privacy is obtainable if on average we change just less than 15\% of a user's rating data before it is released.

We next consider the same tradeoff using the version of our data with the actual ratings. We use k-means clustering with L2 distance, and the results are given in Fig \ref{fig:dist_MI}. We cannot calculate the original mutual information because we do not know the distribution of actual ratings (the number of unique rating vectors is too large compared to the size of our data set), but the mutual information after quantization is 0.182. There is now a much higher average quantization distortion of 37.5, or $0.75$ per rating, which can be accounted for by the fact that the range of values for each show is now 0-5 instead of 0 or 1. We see that using the actual ratings requires slightly higher distortion to reach perfect privacy than with binarized ratings. In this case, we are able to achieve perfect privacy with an extra L2 distortion of about $0.05$ per rating in average.

\noindent\textbf{Inference defeat}: The previous plots show the reduction in privacy leakage that is achieved by our distortion. Another key performance metric is to examine how much the accuracy of a Democrat/Republican classifier is reduced when distorted user ratings are used instead of the non-distorted ones. We consider the example of a logistic regression classifier to infer political affiliation (similar to the one used in \cite{Weinsberg-RecSys-2012} to infer gender from movie ratings). We used 10-fold cross validation on our full dataset, considered both cases of actual and binarized ratings, and a distortion that achieves perfect privacy ($I=0$). After perturbing the ratings to reach $I=0$, any inference algorithm cannot perform better than an uninformed guess. In Fig.~\ref{fig:logregROC} we plot the false positive rate, the number of Democrats falsely classified as Republicans, against the true positive rate, the number of Republicans who are correctly classified. With a distortion bound of $\Delta=1$, we see that we can significantly reduce the classifier's performance but not yet reach perfect privacy; however with $\Delta=2$ the classifier is reduced to nothing more than an uninformed  classifier. This demonstrates that our approach can indeed successfully render inference attempts useless.

Finally, note that logistic regression also performs almost equally well with binarized and actual ratings, which means merely perturbing existing ratings is not enough. The adversary can ignore the actual rating values,  consider only binarized ratings, and classify almost equally well on whether or not a user rated a show. Therefore, we must add and/or delete ratings to protect privacy.

\noindent\textbf{Recommendation quality}: As a final performance metric, we consider the impact of our distortion on the recommendations that would be produced by a recommender system based on matrix factorization. RMSE1 captures the root mean squared error in predicted ratings (compared to the true ratings) using unperturbed data $\hat{r}$, while RMSE2 captures the errors when ratings are predicted using the distorted data $\hat{\hat{r}}$ produced by our algorithm. The results are show in Table \ref{tab:RMSEs}, and were produced using 5-fold cross validation and randomly removing $10\%$ of the ratings in each test set. We can see that any additional errors in TV recommendations, introduced by using our distorted version of user ratings, is small. This preliminary result on the impact on a recommendation system is encouraging, yet requires further extensive testing.

\begin{table}[t!]
\caption{RMSEs of $|r-\hat{r}|$ and $|r-\hat{\hat{r}}|$}
\label{tab:RMSEs}
\centering
\begin{tabular}{c*{6}{p{.8cm}}}
\hline
Set & 1 & 2 & 3 & 4 & 5 \\
\hline\hline
RMSE1 & 1.2506 & 1.1820 & 1.2461 & 1.2155 & 1.2101 \\
\hline
RMSE2 & 1.6972 & 1.6763 & 1.6215 & 1.7248 & 1.8036 \\
\hline
\end{tabular}
\end{table}

\begin{figure*}[tb!]
\centering
\begin{tabular}{| c | l  l  l || l  l  l |}
\cline{2-7}
\multicolumn{1}{c}{} & \multicolumn{3}{| c | }{Initial Cluster} &\multicolumn{3}{| c |}{Privacy Mapping} \\ \hline
N & Family Guy & NFL & Dexter & Family Guy & NFL & Dexter \\ \hline
R & FOX News & NFL & O'Reilly Factor & Daily Show & NFL & Colbert Report \\ \hline
D & Daily Show & Colbert Report & Family Guy & NFL & Dexter & Family Guy  \\ \hline
D & Daily Show & Colbert Report & Dexter & Daily Show & Colbert Report & Dexter \\ \hline
N & Modern Family & The Big bang theory & CNN & Modern Family & The Big bang theory & CNN \\ \hline

\end{tabular}
\caption{Some privacy mappings from clusters to clusters. Each row is a cluster by the 3 most seen TV shows for people within that cluster. Initially, some cluster may be highly correlated with a political affiliation (denoted by D and R), or may be more neutral (denoted by N) in the sense that the distribution of democrats and republicans in the cluster is close to the base distribution in the dataset.}
\label{tab:politics}
\end{figure*}

\section{Related Work}\label{sec:relatedWork}
\subsection{Privacy}\label{sec:relatedPrivacy}

The prevalent notion of privacy adopted by the privacy research community is
differential privacy \cite{Dwork-McSherry-2006,dwork_differential_2006}.  In
broad terms, a query over a database is differentially private if small
variations in the entries of the database does not significantly change the
output distribution of the query. This guarantees that it is difficult to
distinguish ``neighboring'' inputs of the database based solely on the
observation of the output.

Differential privacy does not take into account the distribution of the entries
of the database. This makes the formulation mathematically tractable and
simplifies the implementation of differentially private systems. Moreover, differential
privacy is robust against arbitrary side information from the attacker (also called
background knowledge or auxiliary information), which is a property that our
mechanism cannot guarantee as such, even though our recent works seem to suggest great progress on defining the privacy-utility trade-off under side information. However,
differential privacy does not quantify the amount of information that is leaked
from the system. Furthermore, when inputs are correlated, guaranteeing
differential privacy does not necessarily guarantee \textit{de facto} privacy.
As shown in \cite{alerton2012privacy}, an
adversary might able to infer with arbitrarily high precision the input database
of a differentially private query for certain input distributions.

More general and flexible frameworks similar to differential privacy exist such as the Pufferfish framework \cite{Kifer:2012}. The basic idea behind this framework is to output a pair of mutually exclusive statements, such that the adversary does not know which, if either, of the two statements is true. This framework does not take into account or try to minimize distortion of the data, instead focusing on extracting the data that they wish to keep private, and ignoring utility preservation.  In our paper, we focus on the privacy-utility trade-off. 
We also assume, fairly rigorously, that the adversary has knowledge of the data generation process, and knows the same a priori distribution that we do (which is not necessarily the true distribution). The Pufferfish framework can accommodate any assumption about the adversary's knowledge of the a priori distribution, but also requires that we know what the adversary's belief of the distribution is, which is not knowledge we can assume.

Another existing trend in the privacy research community is to apply
information-theoretic tools to quantify and design privacy-preserving
mechanisms
\cite{evfimievski_limiting_2003,Reed-1973,Yamamoto-ITtrans1983,Sankar-Poor-IFStrans2013,
alerton2012privacy, rebollo-monedero_t-closeness-like_2010}. Information theory
provides a natural framework to measure the amount of private information that
an adversary can learn by observing a given user's public data. This was first
noted by Reed \cite{Reed-1973}, and has since appeared in different forms in the
information theory and privacy literature. One line of work, adopted
in \cite{Yamamoto-ITtrans1983,Sankar-Poor-IFStrans2013}, provides
asymptotic and fundamental limits for an adversary's average equivocation
of the private data as the number of data samples grows arbitrarily
large and characterize rate-distortion-equivocation regions.

Non-asymptotic approaches to information-theoretic privacy were discussed, for
example, in \cite{evfimievski_limiting_2003,
alerton2012privacy, rebollo-monedero_t-closeness-like_2010}. In
\cite{evfimievski_limiting_2003}, information-theoretic metrics were directly
applied to design privacy-preserving mechanisms without considering distortion
constraints. Afterwards,
\cite{rebollo-monedero_t-closeness-like_2010} presented a formulation for
designing privacy-preserving mechanisms similar to the ones found in
rate-distortion theory. More recently, \cite{alerton2012privacy} introduced a
general framework for privacy against statistical inference that takes into
account distortion constraints for the user's public data.

Information-theoretic approaches have also been used to quantify the information flow in security systems (e.g. \cite{hamadou_reconciling_2010} and the references therein). In this case, different information-theoretic metrics are used to quantify the change of an attackers belief of the input of a system given an observation of the output. These approaches, such as the one used in \cite{hamadou_reconciling_2010}, also take into account possible prior mismatches and extra knowledge that an attacker might have. Even though in this paper we also use information-theoretic metrics to quantify the change in the attacker's belief, our results are fundamentally different in what they seek to accomplish. Our main goal is not to simply quantify the adversarial threat, but create a practical framework that allows the design of privacy-preserving mechanisms that also maintain a certain level of utility of the data. Therefore, we simultaneously consider the utility of the data and the variation of the adversary's belief, instead of focusing solely on the information flow. \\ 

\subsection{Quantization}
\label{sec:related quantization}

Data quantization \cite{gray98quantization} are methods that reduce the size of datasets. In summary, all of these methods select $k$ representative examples from the set of $n$ examples, where $k \ll n$. The difference between the methods is in their objectives. One of the most popular methods is $k$-means clustering, which minimizes the mean squared error between the examples and their closest representative example \cite{gray98quantization}. Another popular metric is to minimize is the maximum distance between the example and its closest representative example. Online $k$-center clustering \cite{charikar97incremental} and cover trees \cite{beygelzimer06cover} find nearly optimal solutions to this problem.

\section{Conclusion}\label{sec:Conclusion}

Privacy attacks are receiving more and more attention, both from a theoretical perspective, and from a practical point of view. The amount of information shared everyday, and the recent improvements in inference models have brought in the attention of all, the urge for effective yet private systems. This fundamental contradiction is the core of the privacy problem. In this paper, we show a practical approach to privacy that has roots in a deep and strong theoretical framework. We show that is possible to have private systems by adding a layer of privacy, without changing the way the data is processed afterwards, or its purpose. Using techniques from different fields, such as rate distortion theory, convex optimization, estimation ,and quantization, we address some challenges introduced by the diversity and complexity of real world data. Namely we show that a mismatched prior estimation does not hurt too much in terms of distortion and privacy leakage. Moreover, we propose a generic methodology to deal with big data through quantization. We show that the error in distortion grows linearly in the quantization error, and that the privacy leakage is identical.

\bibliographystyle{IEEEtran}
\bibliography{IEEEabrv,References}

\end{document}